\documentclass[12pt]{article}
\usepackage{natbib}
\usepackage{verbatim}
\usepackage{microtype}
\usepackage{enumitem}
\usepackage{booktabs} 
\usepackage{color,graphicx}
\usepackage{comment}
\usepackage{amssymb,amsmath,dsfont,bm,mathtools,amsthm, bbm}
\usepackage{thmtools, thm-restate}
\usepackage{hyperref,xcolor}

\newcommand{\balance}{\textsf{Balance}}

\newtheorem{theorem}{Theorem}

\newtheorem{lemma}[theorem]{Lemma}

\usepackage{multirow}
\usepackage{booktabs}
\usepackage{subcaption}
\usepackage{adjustbox}

\usepackage{arxiv}
\usepackage[utf8]{inputenc} 
\usepackage[T1]{fontenc}    
\usepackage{amsfonts}       
\usepackage{nicefrac}       
\usepackage{microtype}      
\usepackage[linesnumbered,ruled,vlined]{algorithm2e}
\SetKwInput{KwInput}{Input}                
\SetKwInput{KwOutput}{Output}  

\author{{ Shrikant Saxena } 
		\\
	Department of Computer Science\\
	Indian Institute of Technology (IIT) \\
	Ropar, India \\
	\texttt{2017csb1111@iitrpr.ac.in} \\
	\And
	{Shweta Jain} \\
	Department of Computer Science\\
	 Indian Institute of Technology (IIT) Ropar, India\\
	\texttt{shwetajain@iitrpr.ac.in} 
}
\title{Exploring and Mitigating Gender Bias in Recommender Systems with Explicit Feedback}

\begin{document}
\maketitle
\begin{abstract}
  Recommender systems are indispensable because they influence our day-to-day behavior and decisions by giving us personalized suggestions. Services like Kindle, Youtube, and Netflix depend heavily on the performance of their recommender systems to ensure that their users have a good experience and to increase revenues. Despite their popularity, it has been shown that recommender systems reproduce and amplify the bias present in the real world. The resulting feedback creates a self-perpetuating loop that deteriorates the user experience and results in homogenizing recommendations over time. Further, biased recommendations can also reinforce stereotypes based on gender or ethnicity, thus reinforcing the filter bubbles that we live in. In this paper, we address the problem of gender bias in recommender systems with explicit feedback. We propose a model to quantify the gender bias present in book rating datasets and in the recommendations produced by the recommender systems. Our main contribution is to provide a principled approach to mitigate the bias being produced in the recommendations. We theoretically show that the proposed approach provides unbiased recommendations despite biased data. Through empirical evaluation on publicly available book rating datasets, we further show that the proposed model can significantly reduce bias without significant impact on accuracy. Our method is model agnostic and can be applied to any recommender system. To demonstrate the performance of our model, we present the results on four recommender algorithms, two from the K-nearest neighbors family, UserKNN and ItemKNN, and the other two from the matrix factorization family, Alternating least square and Singular value decomposition.
\end{abstract}



\maketitle
\section{Introduction}
Recommender systems influence a significant portion of our digital activity. They are responsible for keeping the user experience afresh by recommending varied items from a catalog of millions of items and also adapt their recommendations according to the personality and taste of the user. Therefore, a sound recommender system may go a long way in improving user experience quality, hence the user retentivity of a digital outlet.

Recommender systems have historically been judged on their accuracy\cite{herlocker,shani}. When it is concerned with other factors such as novelty, user satisfaction, and diversity\cite{neilHurley, zeigler,knijnenburg}, the focus continues to be just on the satisfaction of the information needs of the users. Although of immense importance to the relevance of a recommender system, these criteria do not capture the complete picture. 
In recent years, the public and academic community have scrutinized artificial intelligence systems regarding their fairness. It has been observed that the results generated by various recommender systems reflect the social biases that exist in human stratum\cite{ekstrand,musicRecommender,fairrecsys}. Scholars have focused on identifying, quantifying, and mitigating the bias present in the results generated by recommendation systems. \citet{multisidedfairness} presents a taxonomy of classes for fair recommendation systems. The author suggests different recommendation settings with fairness requirements such as fairness for only users, fairness for only items, and fairness for both users and items. Our work falls into fairness for only items category where bias is shown by a particular set of users against a specific set of items in the dataset. In particular, we are interested in studying and eliminating users' biasedness against the items associated with a specific gender in recommendation systems.

Bias prevention approaches can be classified according to the phase of the data mining process in which they operate: pre-processing, in-processing, and post-processing methods.  Pre-processing methods aim to control distortion of the training set. In particular, they transform the training dataset so that the discriminatory biases contained in the dataset are smoothed, hampering the mining of unfair decision models from the transformed data. In-processing methods modify recommendation algorithms such that the resulting models do not entail unfair decisions by introducing a fairness constraint in the optimization problem. Lastly, post-processing methods act on the extracted data mining model results instead of the training data or algorithm. The method presented in our work is a hybrid of a pre-processing phase and a post-processing phase.

Two prominent studies have focused on gender bias in recommender systems. The work by \citet{musicRecommender} establishes the existence of bias in the results of the music recommender systems, and the work by \citet{ekstrand} focuses on bias shown by Collaborative Filtering (CF) algorithms while recommending books written by women authors. Both the studies establish that the CF algorithms produced biased results after being fed the biased data from various socio-cultural factors. While both the works focus just on showing the existence of bias in the presence of the users' implicit feedback, we also consider the explicit feedback ratings and the bias that may arise out of it. Thus, our model handles the case when the items associated with specific gender might have received worse feedback than they otherwise ought to achieve by a set of users. We go one step further and propose a model to mitigate these biases by quantifying a particular user's bias and debiasing his or her feedback ratings. We theoretically show that the debiased ratings are unbiased estimators of the true preference of the user. Once the ratings are debiased, they are fed into the recommender algorithms as input to produce recommendations for the desired set of users. Since the recommender system is now fed with the debiased ratings, the resulting recommendations are free from the bias factor and avoid a self-perpetuating loop in the future.

The bias of an individual user reflects his or her taste. However, the KNN based algorithms produce recommendations based on similar characteristics between a set of users and naive implementation of these algorithms reflects the bias of one user in the recommendations produced for the other user. While not directly comparing the rating history of different users or items, Matrix Factorization algorithms rely on deriving latent factors, which depend on the rating history. Both the approaches make the system increasingly biased and homogenized after users interact with their biased recommendations and generate data for the next iteration. The above discussion suggests that though it is necessary to reflect the user's preference in the recommendations produced for him or her to achieve accuracy, it is equally necessary to prevent the bias of one user from reflecting in the recommendations of another similar user. Our research focuses on this particular objective. 

Our debiased ratings assure that the biases of one user do not affect other users; however, it may lead to loss of accuracy because of not reflecting the user's own preferences. We introduce a new step called preference correction which injects the user's preference parameter into his/her own debiased recommendation to maintain the accuracy of the system. On the publicly available Book-Crossing dataset (\cite{bookCrossing}) and Amazon Book Review dataset (\cite{amazon}), we empirically show that this approach retained the significant reduction in bias and had minimal effect on the accuracy of the system. The bias reflected in the recommendations produced by the UserKNN, ItemKNN, ALS, and SVD algorithms is reduced by as much as $42.39\%$, $37.65\%$, $26.51\%$, and $41.43\%$ respectively for the Amazon dataset and by $37.82\%$, $30.73\%$, $24.99\%$, and $32.34\%$ for the Book-Crossing dataset. When measured with respect to Root Mean Squared Error(RMSE), the final accuracy loss in the case of the Amazon dataset comes out to be $7.8\%$, $11.96\%$, $12.49\%$, and $10.38\%$ respectively for the four algorithms. In the case of the Book-Crossing dataset, the RMSE loss comes out to be $13.86\%$, $18.13\%$, $11.41\%$, and $12.89\%$ respectively. In particular, the following are our main contributions.

\subsection{Contributions}
\begin{itemize}
    \item We propose a model to quantify the gender bias in the recommender system when explicit feedback is present.
    \item We propose a principled approach to debias the ratings given and theoretically show that the debiased ratings represent the unbiased estimator of the true preference of the user.
    \item We empirically evaluate our model on publicly available book datasets and show that the approach significantly reduced the biasedness in the system. To show the generality of our proposed approach, we show the results on four algorithms, UserKNN, ItemKNN, ALS, and SVD.
    \item In order to further enhance the accuracy of the debiased system, we propose an approach of preference correction that respects the user's own preferences towards his/her recommendations. We show that the final recommender system significantly reduces the bias in the system while not deteriorating the accuracy much.
\end{itemize}

\section{Related Works}
The problem of gender bias and discrimination has received lots of attention in recent works \cite{datamining11}. Many proposals like \cite{datamining29}, \cite{datamining28}, \cite{datamining31}, \cite{datamining25}, \cite{datamining26}, \cite{datamining30} are dedicated to detecting and measuring the existing biases in the datasets while other efforts \cite{datamining18}, \cite{datamining19}, \cite{datamining12}, \cite{datamining13}, \cite{datamining14}, \cite{datamining5}, \cite{datamining33} are focused on ensuring that data mining models do not produce discriminatory results even though the input data may be biased. Most of these works focus on the classical problem of classification. 
\citet{dataminingCriticisation} discuss the application of various classification methods like Support Vector Machines, Artificial Neural Networks, Bayesian classifiers, and decision trees in recommender systems. Their findings indicated that a more complex classifier need not give a better performance for recommender systems, and more exploration is needed in this direction. 

When considering "fairness for only users" according to the taxonomy presented by \citet{multisidedfairness}, \cite{fairrecsys} and \cite{metric} discuss the bias with respect to the preferential recommending of certain items only to the users of a specific gender. While weighted regularization matrix factorization studied in \cite{fairrecsys} is only appropriate for implicit feedback, the Group Utility Loss Minimization proposed in \cite{metric} works only with respect to the UserKNN algorithm. Both the papers address the issue of gender bias by employing post-processing algorithms that work only in limited settings. Though  \cite{fairrecsys} and \cite{metric} have addressed the issue of fairness of recommender systems with respect to gender, they have done so from the perspective of recommending certain items only to the users of a specific gender. The difference between their work and our study lies in the fact that we focus on the more direct issue of gender bias in recommendations shown to items associated with a specific gender.

\begin{figure}[]
\includegraphics[scale=0.8]{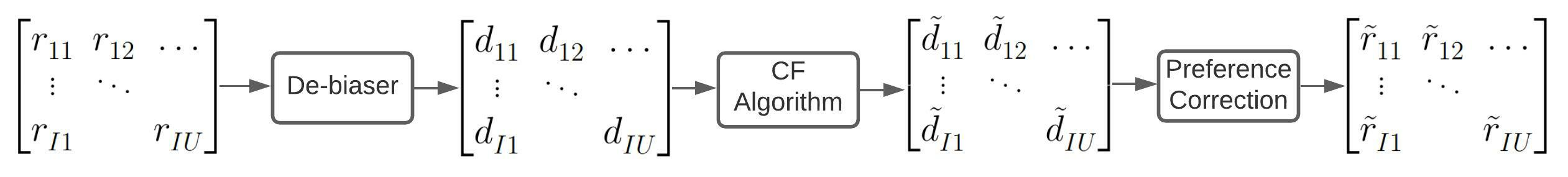}
\caption{Model schematics}
\label{deBiasingReBiasing}
\end{figure}



\citet{musicRecommender} in their research highlight the artist gender bias in music recommendations produced by Collaborative Filtering algorithms. 
The work traces the causes of disparity to variations in input gender distributions and user-item preferences, highlighting the effect such configurations can have on user's gender bias after recommendation generation. 
\citet{potentialFactors}  discuss the biases from the perspective of a specific group of individuals (for example, a particular gender) receiving less calibrated and hence unfair recommendations. 
\citet{ekstrand} explores the gender bias present in the book rating dataset. 
Our work is different from the works by \citet{musicRecommender}, \citet{potentialFactors} and \citet{ekstrand} in primarily two factors: (i) we consider explicit feedback as opposed to the implicit feedback, and (ii) we propose a principled approach to debias the ratings and theoretically show that the debiased ratings are unbiased estimators of true ratings.

The research by \citet{contentBasedFiltering} focuses on algorithmic gender bias and proposes a  framework whereby language-based data may be systematically evaluated to assess levels of gender bias prevalent in training data for machine learning systems. 
Our work is different from this study as this study is focused on evaluating gender bias in the language and textual data settings, while ours deals with gender bias in a more traditional user-item rating setting.
\section{The Model}
Consider a recommender system having $\mathcal{U} = \{1,2,\ldots,U\}$ users and $\mathcal{I} = \{1, 2, \ldots, I\}$ items. Let $\mathbb{D}$ and $\mathbb{A}$ denote the set of items associated with disadvantaged group and advantaged group, respectively. For example, in a book recommender system, the books represent the items; $\mathbb{D}$ and $\mathbb{A}$ represent the set of books written by women and men authors respectively.  With respect to book recommender system, researchers have already shown that the data is biased against female authors' books \cite{ekstrand}. 

Let $r_{ui} \in [1,R]$ denote the rating that user $u$ has given to the item $i$. As opposed to previous works, we consider explicit feedback wherein biases may not only arise from not giving the rating to the item but may also come from giving a bad rating to the item. The user profile $p_u = \{X_u, R_u\}$ represents the set of books ($X_u$) and the ratings ($R_u = \{r_{ui}\}_{i\in X_u}$) that user $u$ has given to those items. 

The proposed recommender system first pre-processes the data that: 1) finds the log-bias $\theta_u$ of each user $u$ and 2) generates the debiased rating $d_{ui}$ of each user $u$ and item $i$ using the computed bias in the first step. We then theoretically show that the debiased ratings generated are unbiased estimators of the true preferences of the user for the items rated by them. Thus, the debiased dataset can then be fed into various recommender algorithms to generate an unbiased predicted rating of a user $u$ for the item $i$, denoted by $\tilde{d}_{ui}$. This debiasing step ensures that the existing biases are not boosted further in the system. Our debiasing model is independent of any recommendation algorithm. We show the performance of our debiasing model on both K-nearest neighbors-based algorithms (UserKNN, ItemKNN) as well as matrix factorization-based algorithms (Alternating Least Square and  Singular Value Decomposition) to produce the recommendations. 

In the next step, we use preference corrector to reintroduce the preferences of a particular user $u$ to his/her own recommendations. This is achieved via producing a user specific rating $\tilde{r}_{ui}$ from the debiased rating $\tilde{d}_{ui}$. The recommendations are re-ranked according to the adjusted ratings, and the recommendations are presented to the user. This step ensures that the system does not lose accuracy for not considering the preferences of the users. Figure \ref{deBiasingReBiasing} shows the schematic diagram of our model. Consider that the ratings $r_{ui}$ are continuous values ranging from $1$ to $R$, then mathematically, a biased recommender system can be represented as follows:
\begin{enumerate}
    \item Each user $u$, while rating an item $i$, scales down the maximum rating $R$ by $e^{p_{ui}}$. $p_{ui}$ is a random variable, drawn from a distribution function $P_u(I)$, which has a mean value of $\alpha_u$. $p_{ui}$ represents the logarithm of the true preference of the user $u$ for the item $i$. For the sake of brevity, we call it log-preference of the user $u$ for the item $i$. Hence $e^{p_{ui}}$ is a representation of the true preference of user $u$ for the item $i$. 
    \item In case the item is associated with the disadvantaged group, the user $u$ further scales down the rating of the item by a factor $e^{q_{ui}}$. $q_{ui}$ is a random variable, drawn from a distribution function $Q_u(I)$ having a mean value of $\beta_u$. $q_{ui}$ represents the logarithm of the biasedness of the user $u$ shown to the item $i$. For the sake of brevity, we call it the log-bias of the user $u$ for the book $i$. Hence $e^{q_{ui}}$ represents the biasedness of the user $u$ for the book $i$.
    \item For each user $u$, $\beta_u$ is sampled from the a distribution function $\Omega(x)$ which governs the global log-bias tendency of the users. We denote the mean value of $\Omega(x)$ by $\gamma$.
\end{enumerate}
Thus, ratings $r_{ui}$ can be expressed as:

\begin{equation}
\label{ratingEquation}
  r_{ui} = 
  \begin{cases}
  R/e^{p_{ui}},& \textit{if $i$ is associated with advantaged group} \\
  R/e^{p_{ui}} e^{q_{ui}},& \textit{if $i$ is associated with disadvantaged group}
  \end{cases}
\end{equation}
We now present a detailed description of each of the step. 
\subsection{Estimating the mean value for log-bias}
\label{inputBias}
The geometric mean of the ratings given by a user $u$ to the items associated with disadvantaged and advantaged groups, denoted by $r_{ud}$ and $r_{ua}$ respectively, are given by the following expressions: 
\begin{align*}
r_{ud} = \left(\prod_{i \in \mathbb{D} \cap X_u}r_{ui}\right)^{1/|\mathbb{D} \cap X_u|}\ \text{ and }\ r_{ua} = \left(\prod_{i \in \mathbb{A} \cap X_u}r_{ui}\right)^{1/|\mathbb{A} \cap X_u|}
\end{align*}
Further, the log bias in the user profile $p_u$, is given by
    $\theta_u = \ln\left(\frac{r_{ua}}{r_{ud}}\right)$.

We use geometric mean to compute the average rating of a user due to the following reasons: 1) It is less biased towards very high scores as compared to arithmetic mean\cite{neve2019latent} and 2) when cold users are involved, aggregating recommendations using the geometric mean is more robust as compared to arithmetic mean\cite{valcarce2020assessing}. 

The below lemma shows that $\theta_u$ is an unbiased estimator of $\beta_u$.

 
 \begin{lemma}
 \label{lemma1}
The expectation of log-bias, $\theta_u$ in the user profile $p_u$ represents the mean value of the log-bias, $\beta_u$ of the user $u$.
 \end{lemma}
 \begin{proof}
 Let us denote $m = |\mathbb{D} \cap X_u|$ and  $n = |\mathbb{A} \cap X_u|$ to be the number of items associated with disadvantaged and advantaged group respectively in user profile $p_u$. Then,
 \allowdisplaybreaks
 \begin{align*}
 \theta_u &= \ln\left(\frac{r_{ua}}{r_{ud}}\right) = \ln \left[ \frac{\left ( \prod_{y=1}^m e^{p_{uy}} e^{q_{uy}} \right)^{\frac{1}{m}}}{\left ( \prod_{x=1}^n e^ {p_{ux}} \right)^{\frac{1}{n}}} \right ]\tag*{(Using equation \ref{ratingEquation})}\\
    &=\frac{1}{m}\sum_{y=1}^m q_{uy} + \frac{1}{m}\sum_{y=1}^m p_{uy} - \frac{1}{n}\sum_{x=1}^n p_{ux}
\end{align*}
Taking expectation both sides:
\begin{equation}
\label{simplifiedRatingEquation}
    \mathbb{E}[\theta_u] = \mathbb{E}\left[\frac{1}{m}\sum_{y=1}^m q_{uy} + \frac{1}{m}\sum_{y=1}^m p_{uy} - \frac{1}{n}\sum_{x=1}^n p_{ux}\right] 
\end{equation}

Using linearity of expectation and some simplification, we get: 
\begin{align*}
    \mathbb{E}[\theta_u] &= \frac{1}{m}\sum_{y=1}^m \mathbb{E}[q_{uy}] + \frac{1}{m}\sum_{y=1}^m \mathbb{E}[p_{uy}] - \frac{1}{n}\sum_{x=1}^n \mathbb{E}[p_{ux}]\\
    &=\frac{1}{m}\sum_{y=1}^m \beta_u + \frac{1}{m}\sum_{y=1}^m \alpha_{u} - \frac{1}{n}\sum_{x=1}^n \alpha_u
\end{align*}
Thus, $\mathbb{E}[\theta_u] = \beta_u$.
 \end{proof}
Once we get the log biasedness tendencies of users, we use them to produce the debiased ratings for the given dataset. 
\subsection{Debiasing the Dataset}

The debiased rating of the item $i$ associated with disadvantaged group and rated by user $u$ is given as $d_{ui} = r_{ui}e^{\theta_u}$
We now provide the main theorem of our paper.
\begin{theorem}
$\ln(d_{ui})$ is the unbiased estimator of the log of the true rating of the item $i$.
\end{theorem}
\begin{proof} $\ln(d_{ui}) = \theta_u + \ln(r_{ui}) =\theta_u + \ln R - p_{ui} - q_{ui}$. Last equality is obtained from Equation \ref{ratingEquation}. 
Taking expectation both sides:
\allowdisplaybreaks
\begin{align*}
    \mathbb{E}(\ln(d_{ui})) &= \mathbb{E}[\theta_u]+\mathbb{E}[\ln R] - \mathbb{E}[q_{ui}] - \mathbb{E}[p_{ui}]\\
    &=\beta_u + \ln R - \beta_u - \alpha_u \tag*{(Using Lemma \ref{lemma1})}\\
    &= \ln R - \alpha_u = \ln \left( \frac{R}{e^{\alpha_u}} \right) 
\end{align*}
 
As we can see, the expected value of $\ln(d_{ui})$ contains only the term representing the true preference of the item for user $u$. 
\end{proof}
Thus, instead of $r_{ui}$, ratings $d_{ui}$ are fed into the recommender system to generate the predicted unbiased ratings $\tilde{d}_{ui}$. Simply removing the bias from the user's rating could severely affect the system's accuracy because the bias of an individual user reflects their taste. However, the debiasing step helps prevent the bias of one user from affecting the recommendation of other users. Next, we use preference corrections by correcting the predicted rating of the user with respect to his/her own preference parameter. 
\subsection{Preference Correction to Improve Accuracy}
\label{reintroducingBias}

Note that when the users are inherently biased against a group of items, $\mathcal{D}$ then showing the items from $\mathcal{D}$ naively to these users will severely affect the accuracy of the system. The goal of this work is not just to promote the exposure of the items among the two groups but is to not let the bias of one user creep into the bias of the other user. This was achieved via debiasing the dataset. Once the debiased ratings are generated, the accuracy of the system is maintained by introducing a correction factor. Although providing us with higher accuracy, the idea to re-introduce the correction factor may lead to an overall increase in the individual biases. This on a prima-facie may look self-defeating, but we need to note that final ratings still have significantly less bias than original ratings. If we do not introduce the correction factor, the users might flock to a substantial bias platform due to poor accuracy.

The correction is achieved via multiplying the predicted ratings of items associated with disadvantaged group by a factor $e^{-\theta_{u}}$. Thus, the final recommended ratings will be given as $\tilde{r}_{ui} = \tilde{d}_{ui}e^{-\theta_{u}}$. 
Similar to the calculation of bias in the dataset, we can now compute the bias in the recommendation profile.

\subsection{Bias in recommendation profile}
\label{outputBias}

We generate recommendations for the users in the test set $\mathcal{T}$. The recommendation profile for a user $u \in \mathcal{T}$ is denoted by $\tilde{p}_u = \{\tilde{X}_u, \tilde{R}_u\}$, which represents the set of recommended books ($\tilde{X}_u$) for the user $u$ and their predicted ratings ($\tilde{R}_u = \{\tilde{r}_{ui}\}_{i\in \tilde{X}_u}$).
Let the set of items associated with disadvantaged and advantaged groups be denoted by $\tilde{\mathbb{D}}$ and $\tilde{\mathbb{A}}$ respectively. The average predicted ratings of the items associated with disadvantaged and advantaged groups, denoted by $\tilde{r}_{ud}$ and $\tilde{r}_{ua}$ respectively, are given by:
$\tilde{r}_{ud} = \left(\prod_{i \in \tilde{\mathbb{D}} \cap \tilde{X_u}}\tilde{r}_{ui}\right)^{1/|\tilde{\mathbb{D}} \cap \tilde{X}_u|}\ \text{ and }\
\tilde{r}_{ua} = \left(\prod_{i \in \tilde{\mathbb{A}} \cap \tilde{X_u}}\tilde{r}_{ui}\right)^{1/|\tilde{\mathbb{A}} \cap \tilde{X}_u|}$
where $\tilde{r}_{ui}$ is the predicted rating given to item $i$ in the recommendation-profile generated for a user $u$. The log-bias in the recommendation-profile $p_u$, denoted by $\tilde{\theta}_u$, is then given by 
    $\tilde{\theta}_u = \ln\left(\frac{\tilde{r}_{ua}}{\tilde{r}_{ud}}\right)$.
For an unbiased recommendation-profile, $\tilde{\theta}_u = 0$. A profile biased against disadvantaged groups will have $\tilde{\theta}_u > 0$. We can then compute the overall bias of the recommender system by taking the average overall users, and this average gives us the estimated value of $\gamma$. 




\section{Dataset}
To evaluate the proposed model, we run experiments on two publicly available book rating datasets, the Book-Crossing dataset, originally put together by \citet{bookCrossing} and the Amazon Book Review dataset, put together by \citet{amazon}. We further process this dataset through the following stages: 
\subsection{Book Author Identification}
Their unique ISBNs identify the books in both datasets. We identified the authors of the books present in the datasets via their ISBN numbers using the following three API services: \emph{Google Books API} \cite{googleAPI}, \emph{ISBNdb API} \cite{isbndbAPI}, and \emph{Open Library API} \cite{openLibraryAPI}.
We could not identify the authors of some of the books. Hence we discarded those books from the dataset.
\subsection{Author Gender Identification}
We identified the genders of the authors via their first names. We used \emph{ Genderize.io } \cite{genderize}, an API service dedicated to identifying the gender given the first name of the person. We used a minimum confidence threshold of $90\%$ for gender identification. We could not identify the gender of some of the authors. We discard the books written by those authors from the dataset.

\subsection{Filtering}
We filtered the Book-Crossing dataset to include only those books with at least $50$ ratings and only those users who have rated at least $50$ books. Amazon dataset was significantly larger as compared to the Book-Crossing dataset. We filtered it to include only those books with at least $100$ ratings and only those users who have rated at least $100$ books. We did this filtering so that recommender algorithms have much data to produce accurate recommendations. The statistics of filtered datasets are mentioned in Table \ref{tab:datasetStats}. The number of books written by male authors is almost equal to that of female authors for both datasets.

\begin{table}[]
\centering
\begin{tabular}{|c|c|c|}
\cline{1-3}
Statistic & Amazon & Book-Crossing \\ \cline{1-3}
Number of male authored books & 58369 & 829\\ \cline{1-3}
Number of female authored books & 58220 & 806 \\ \cline{1-3}
Number of users & 44792 & 376\\ \cline{1-3}
\end{tabular}
\caption{Dataset details}
\label{tab:datasetStats}
\end{table}

\section{Experimental Results}
\subsection{Input Bias}

We show the distributions of log-bias tendency ($\theta_u$) of the users in the Amazon dataset and the Book-Crossing dataset in Figure \ref{fig:inputLogBias}. We observe that the mean log-bias tendency over all the users in the Amazon dataset is higher (0.176) than that of the Book-Crossing dataset (0.157)\footnote{code is available at \url{https://github.com/venomNomNom/genderBias.git} }
. 


\begin{figure}[]
\begin{subfigure}{.5\textwidth}
  \centering
  \includegraphics[scale=0.5]{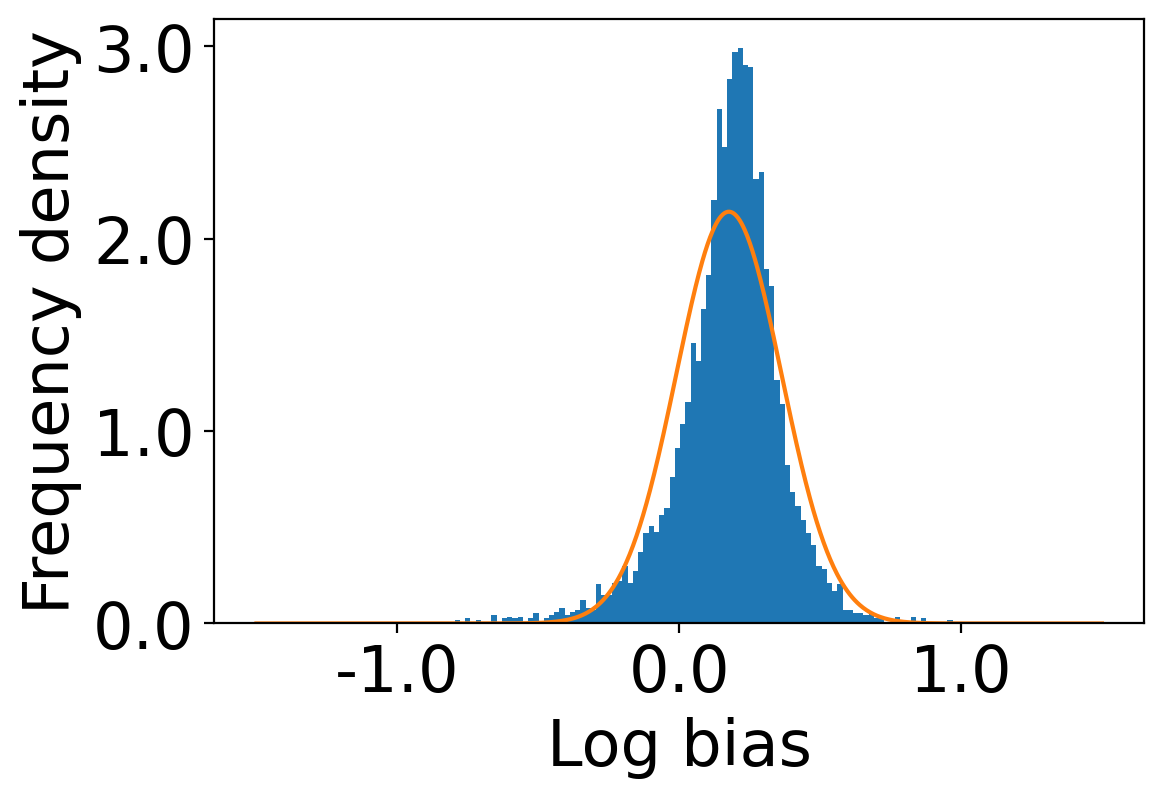}  
  \caption{Amazon dataset}
  \label{fig:inputBiasAZ}
\end{subfigure}
\begin{subfigure}{.5\textwidth}
  \centering
  \includegraphics[scale=0.5]{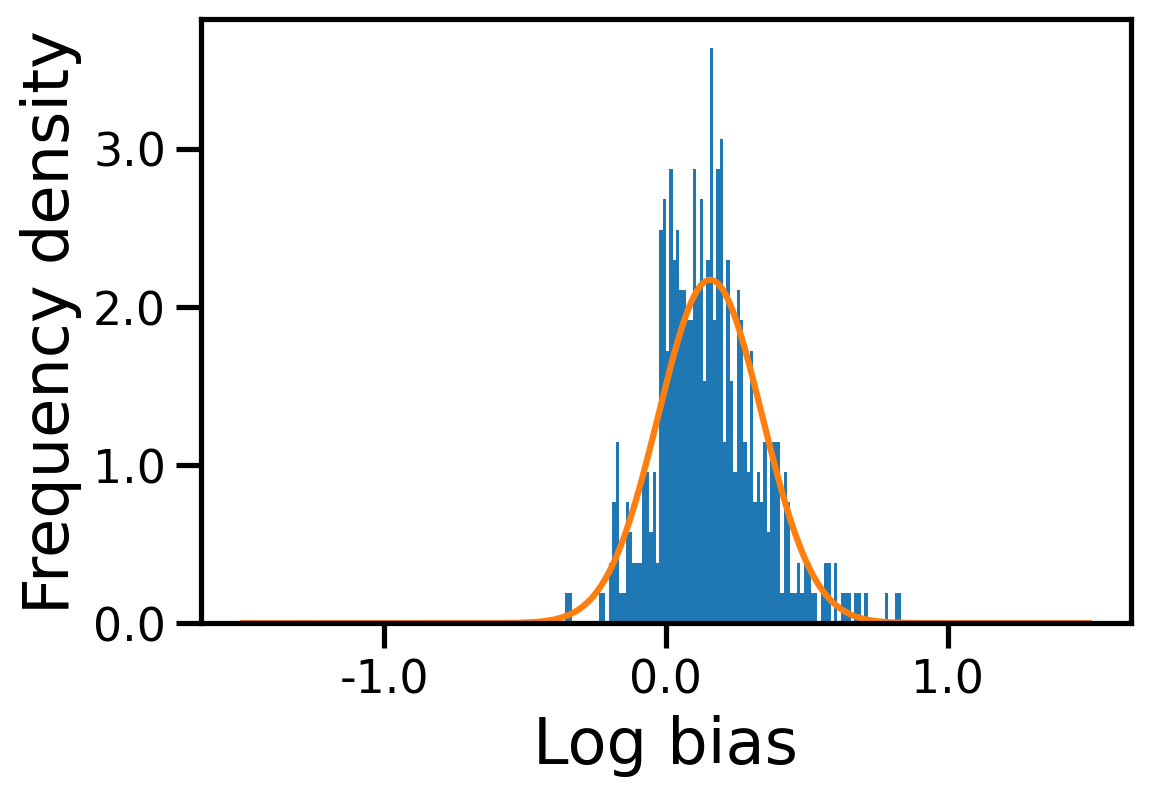}  
  \caption{Book-Crossing dataset}
  \label{fig:inputBiasBX}
\end{subfigure}
\caption{User log-bias in the original dataset}
\label{fig:inputLogBias}
\end{figure}



\subsection{Output Bias}
We randomly separate $20\%$ of users in each dataset as the test group. We generate the recommendations for the users in the test group using two K-nearest neighbors-based algorithms, UserKNN and ItemKNN, and two matrix factorization-based algorithms, Alternating Least Square and Singular Value Decomposition. These algorithms were selected because the accuracy and ranking relevancy of the recommendations produced by them were among the highest values compared with other algorithms. Hence coupling our model with them would best highlight the effects brought about by the same. We calculate the estimated value of log-bias ($\tilde{\theta}_u$) and accuracy in the recommendations separately for each algorithm applied on the two datasets. For this, we use two error measures, the Root Mean Squared Error (RMSE) and the Mean Absolute Error (MAE), and two ranking relevance parameters, Normalized Discounted Cumulative Gains and Mean Reciprocal Rank.

We first begin plotting the log-bias ($\tilde{\theta}_u$) distribution for the recommendations produced by the algorithms without employing our debiased model in Figures  \ref{fig:outputLogBiasWithoutModelAZ} and \ref{fig:outputLogBiasWithoutModelBX} for Amazon and Book-Crossing datasets respectively. We compute the log-bias by feeding biased ratings $r_{ui}$ to the four algorithms. Table \ref{tab:azResults} shows the exact values for comparison with other cases. We next deploy our model partially. We leave out the preference correction phase and produce the recommendations using the algorithms mentioned before by feeding the debiased ratings $d_{ui}$ to these algorithms. We estimate the mean log-bias tendency in the recommendations $\tilde{\theta}_u$ using debiased ratings produced by the algorithms $\tilde{d}_{ui}$. The log-bias ($\tilde{\theta_u}$) distribution for the recommendations produced by the algorithms for the Amazon as well as Book-Crossing dataset after partial deployment of the model is depicted in the Figure \ref{fig:outputLogBiasWithoutAccuracyAZ} and Figure \ref{fig:outputLogBiasWithoutAccuracyBX} with the exact values provided in the Table \ref{tab:azResults}. As can be seen, there is a significant reduction in log-bias tendency ($64.38 \%$) in the Amazon dataset and ($53.67 \%$) in Book-Crossing dataset for the UserKNN algorithm. However, we also see an increase in error rates on both datasets. This is because the test data itself contains biases. 

Finally, we deploy our complete model and repeat the experiment. The log-bias ($\tilde{\theta_u}$) distribution for the recommendations produced by the algorithms after deployment of the complete model is depicted in Figures \ref{fig:outputLogBiasWithAccuracyAZ} and \ref{fig:outputLogBiasWithAccuracyBX} with the values in the Table \ref{tab:azResults}. As can be seen, there is still a significant reduction in mean log-bias tendency, which reduces by $42.39\%$ in the Amazon dataset and by $37.82\%$ in the case of the Book-Crossing dataset for UserKNN algorithm. The accuracy loss, however, is insignificant, making this trade-off advantageous. Figure \ref{fig:biasReduction} presents the percentage gain in bias reduction for both the dataset. The percentage loss in accuracy is depicted in figures \ref{fig:accuracyLossAZ} and \ref{fig:accuracyLossBX} for Amazon and Book-Crossing datasets respectively. The percentage loss in ranking relevancy metrics are depicted in figures \ref{fig:rankingRelevancyLossAZ} and \ref{fig:rankingRelevancyLossBX} respectively.


\begin{figure}[]
\begin{subfigure}{.24\textwidth}
  \centering
  \includegraphics[scale=0.25]{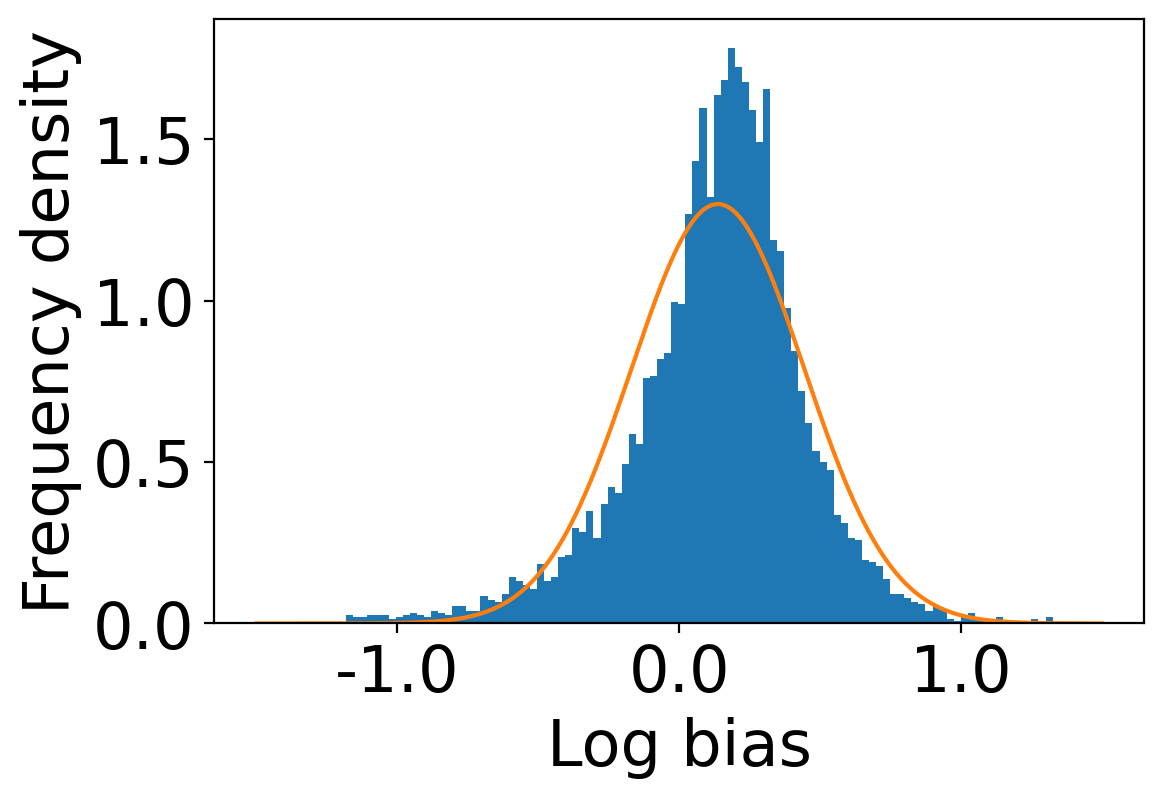}
  \caption{UserKNN}
  \label{fig:outputBiasWithoutModelUserKNNAZ}
\end{subfigure}
\begin{subfigure}{.24\textwidth}
  \centering
  \includegraphics[scale=0.25]{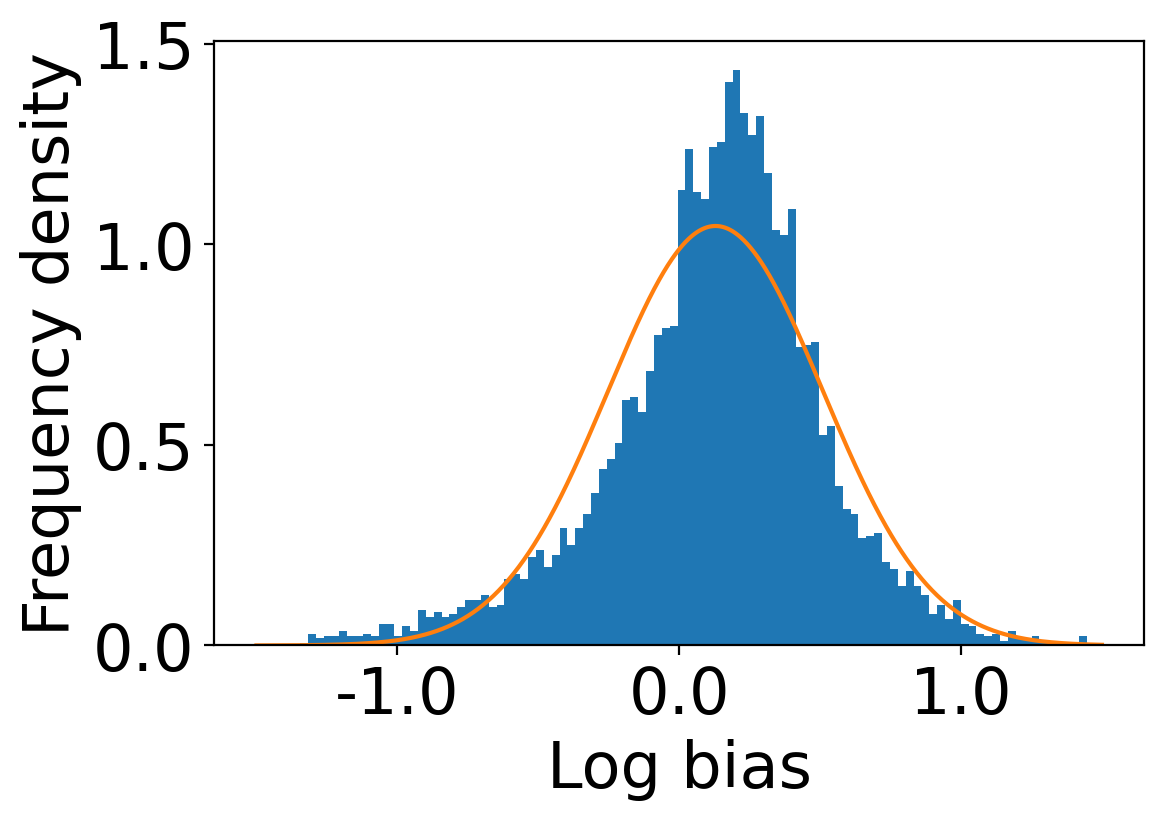}
  \caption{ItemKNN}
  \label{fig:outputBiasWithoutModelItemKNNAZ}
\end{subfigure}
\begin{subfigure}{.24\textwidth}
  \centering
  \includegraphics[scale=0.25]{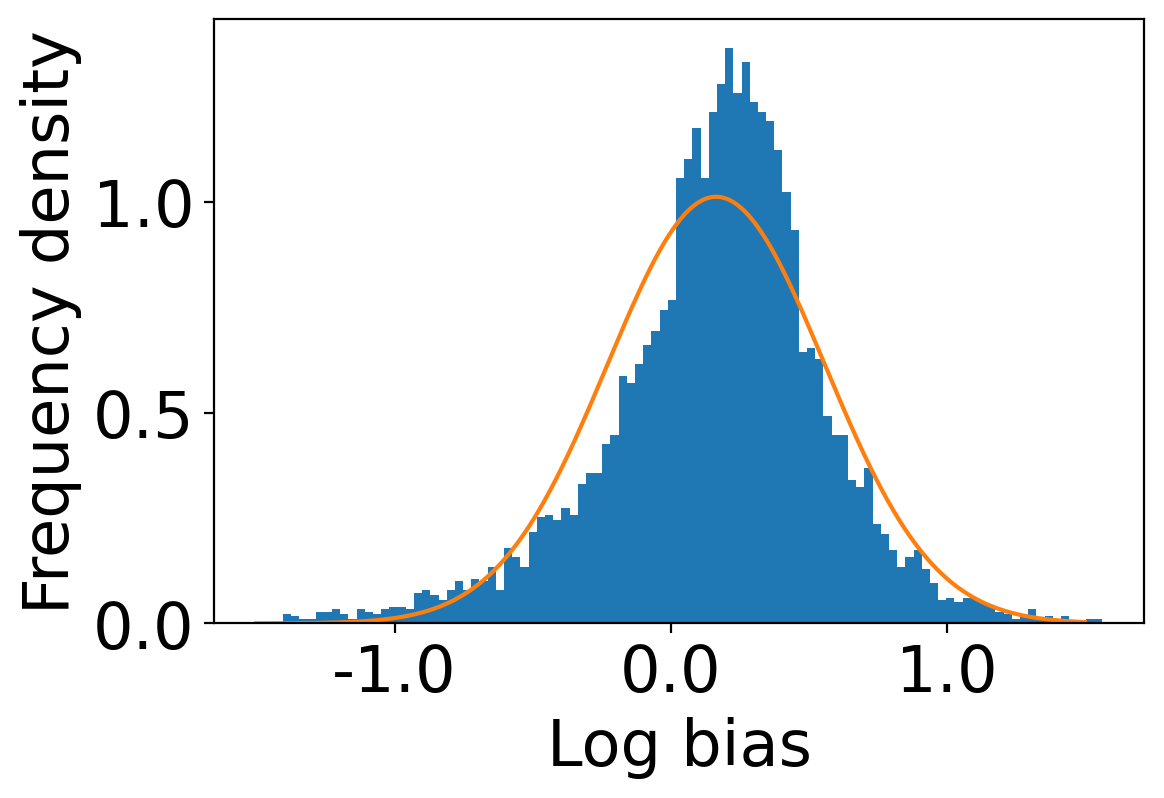}
  \caption{ALS}
  \label{fig:outputBiasWithoutModelAlsAZ}
\end{subfigure}
\begin{subfigure}{.24\textwidth}
  \centering
  \includegraphics[scale=0.25]{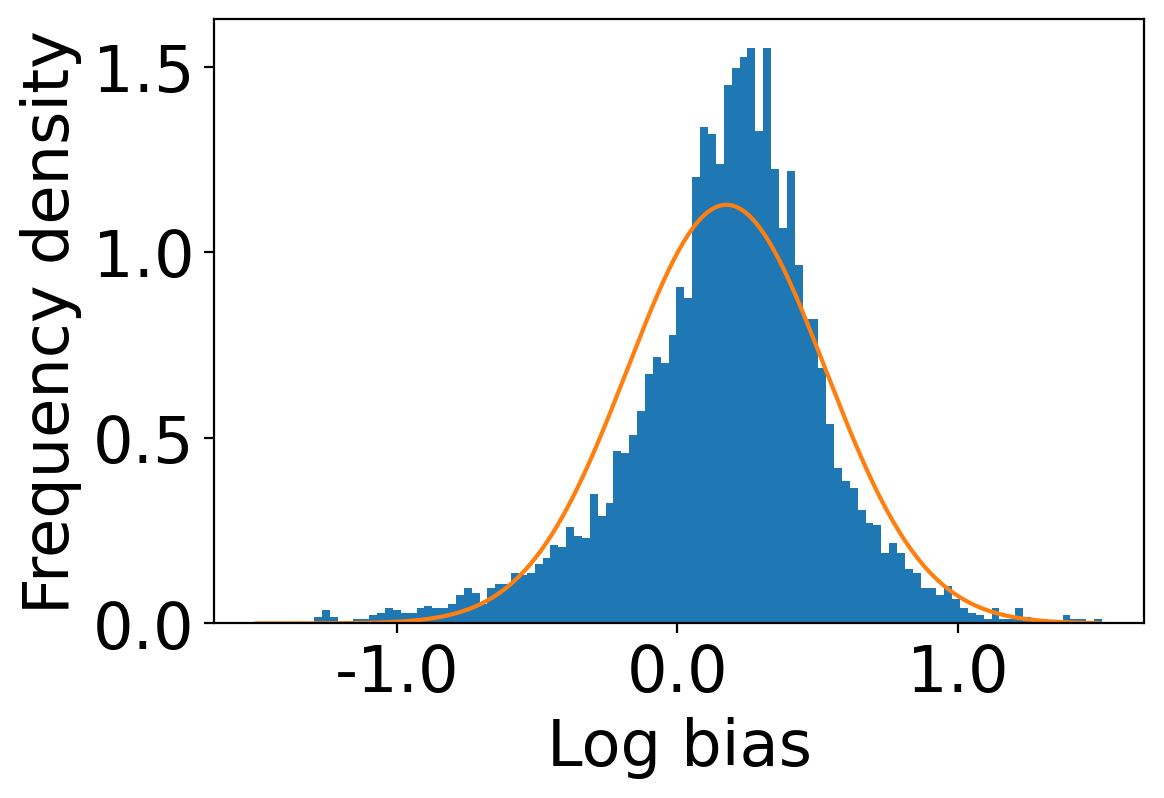}
  \caption{SVD}
  \label{fig:outputBiasWithoutModelSvdAZ}
\end{subfigure}
\caption{Output log-bias in AZ dataset without employing the model}
\label{fig:outputLogBiasWithoutModelAZ}
\end{figure}


\begin{figure*}[ht!]
\begin{subfigure}{.24\textwidth}
  \centering
  \includegraphics[scale=0.25]{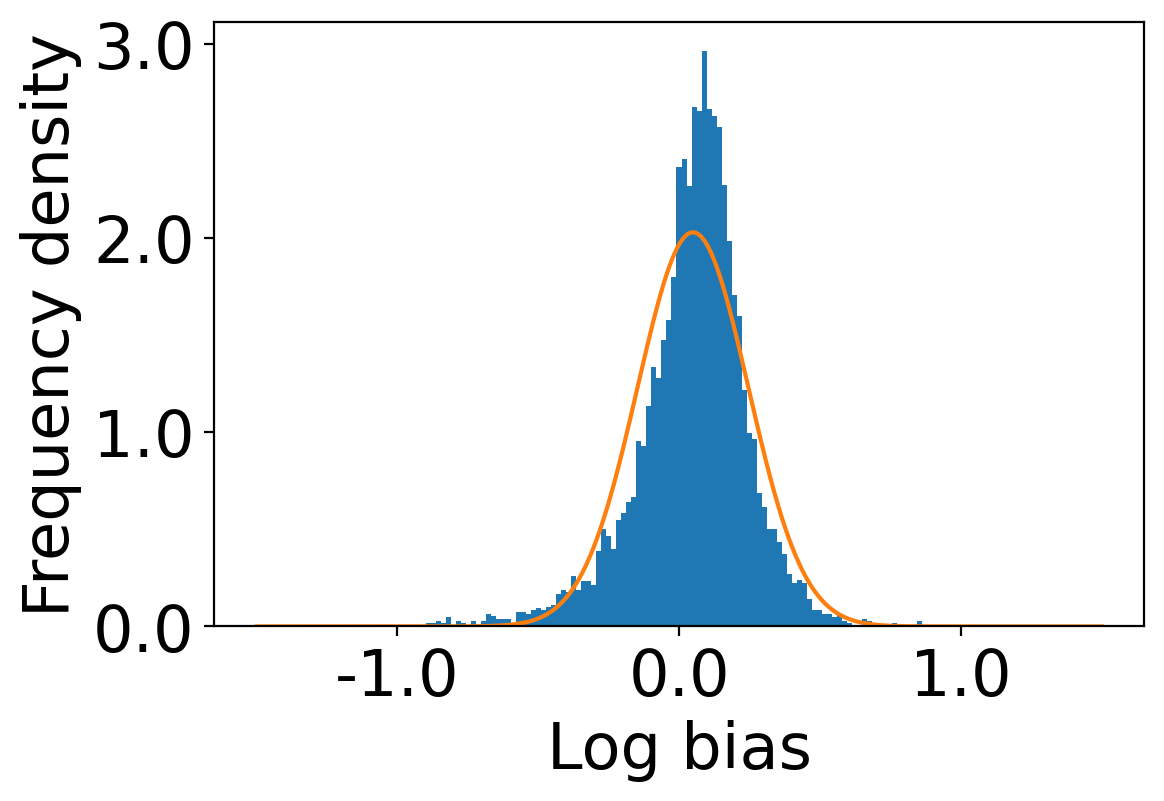}
  \caption{UserKNN}
  \label{fig:outputBiasWithoutAccuracyUserKNNAZ}
\end{subfigure}
\begin{subfigure}{.24\textwidth}
  \centering
  \includegraphics[scale=0.25]{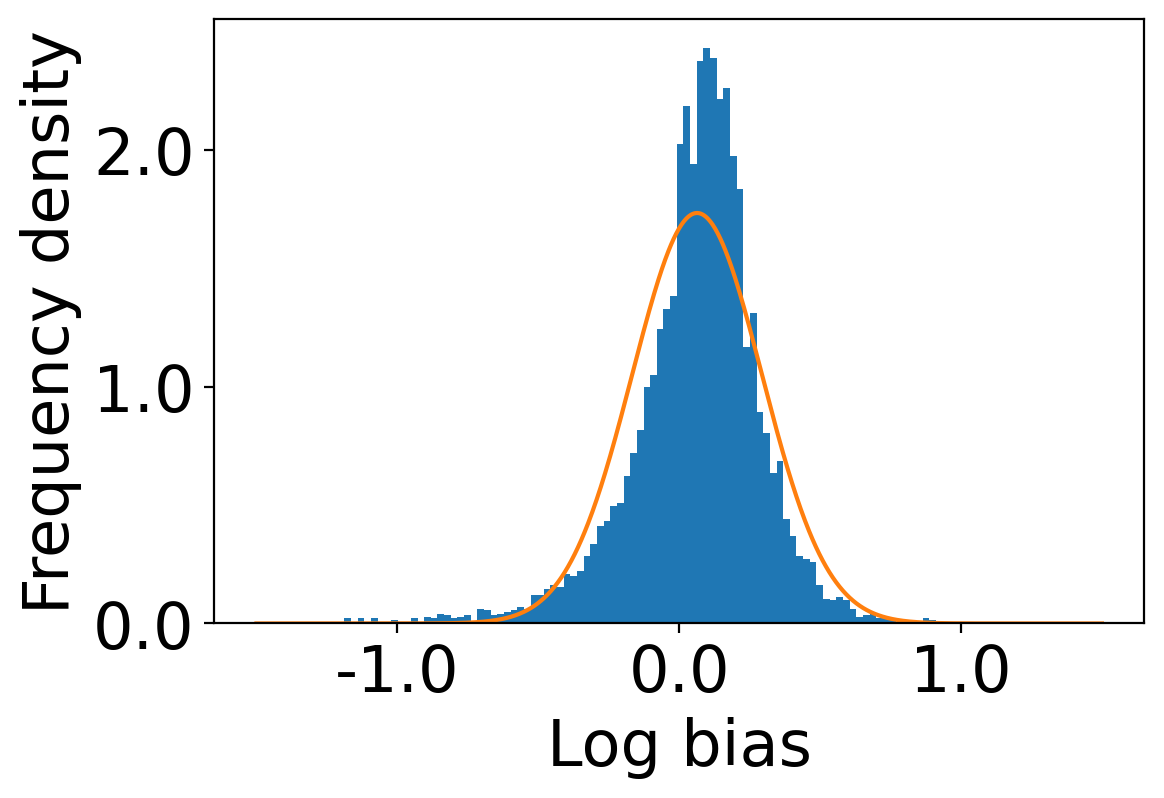}
  \caption{ItemKNN}
  \label{fig:outputBiasWithoutAccuracyItemKNNAZ}
\end{subfigure}
\begin{subfigure}{.24\textwidth}
  \centering
  \includegraphics[scale=0.25]{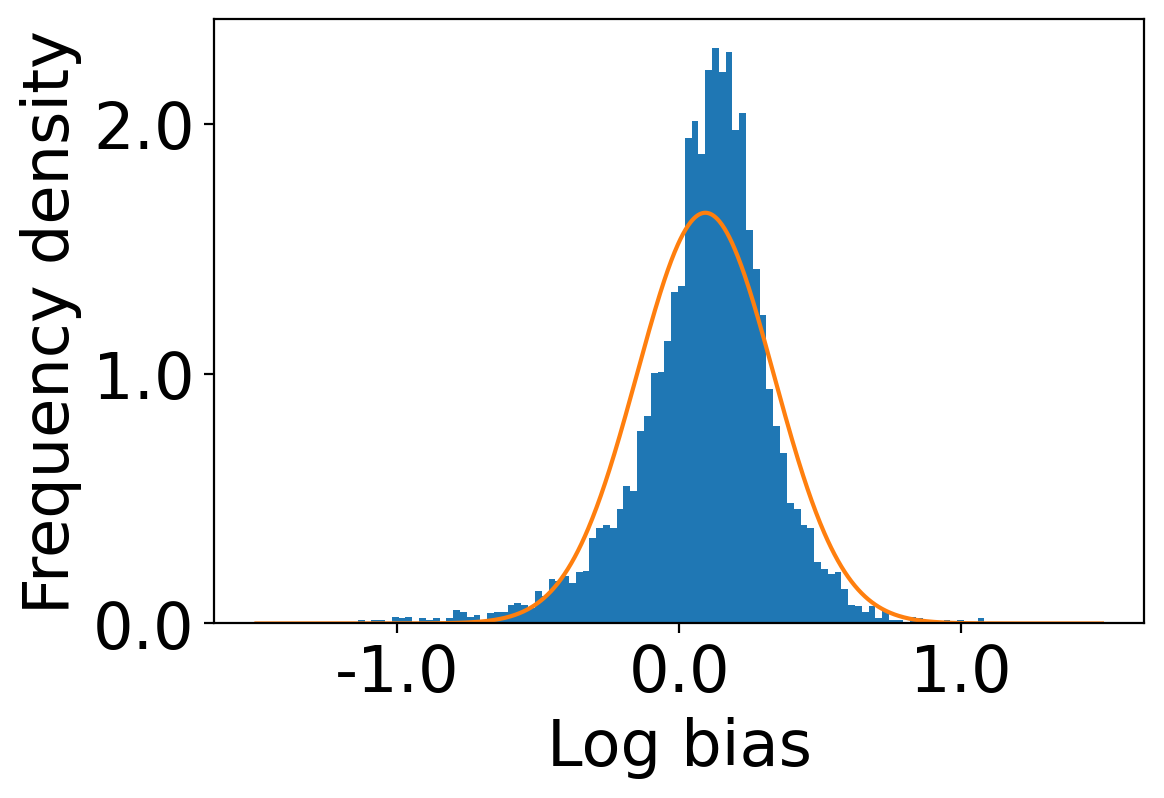}
  \caption{ALS}
  \label{fig:outputBiasWithoutAccuracyAlsAZ}
\end{subfigure}
\begin{subfigure}{.24\textwidth}
  \centering
  \includegraphics[scale=0.25]{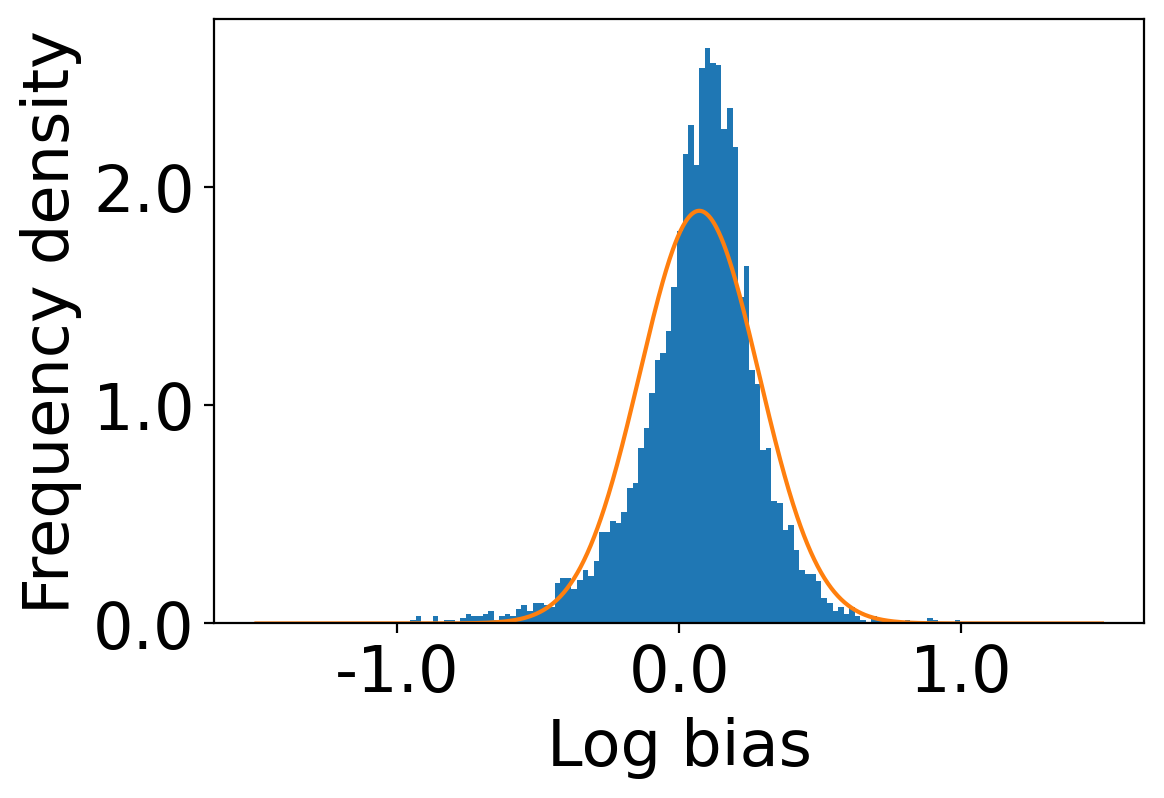}
  \caption{SVD}
  \label{fig:outputBiasWithoutAccuracySvdAZ}
\end{subfigure}
\caption{Output log-bias in AZ dataset with debiasing}
\label{fig:outputLogBiasWithoutAccuracyAZ}
\end{figure*}


\begin{figure*}[]
\begin{subfigure}{.24\textwidth}
  \centering
  \includegraphics[scale=0.25]{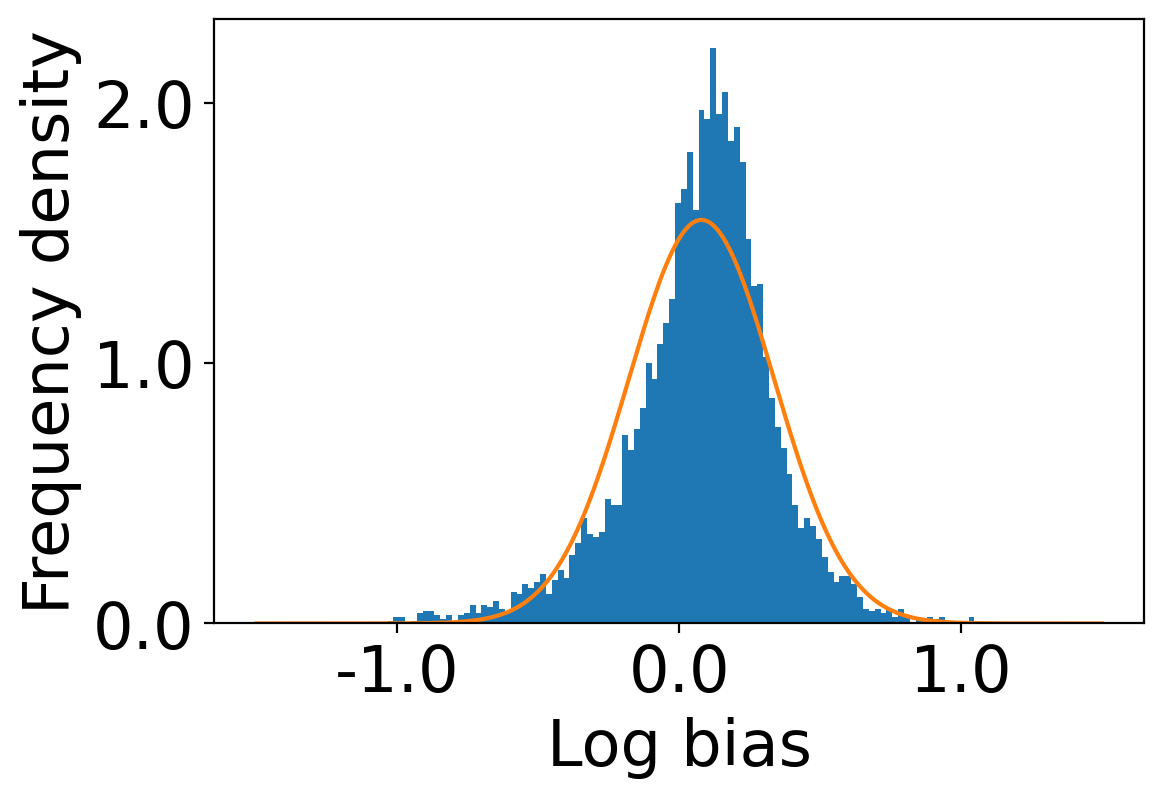}
  \caption{UserKNN}
  \label{fig:outputBiasWithAccuracyUserKNNAZ}
\end{subfigure}
\begin{subfigure}{.24\textwidth}
  \centering
  \includegraphics[scale=0.25]{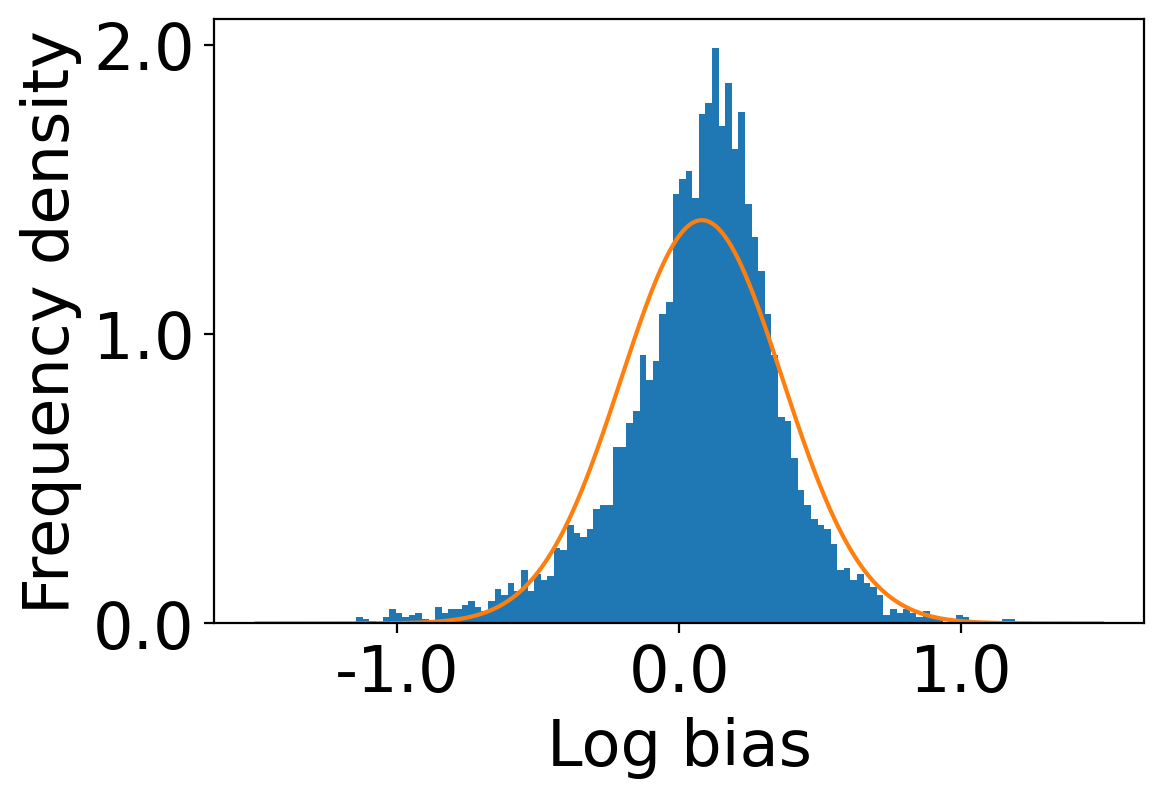}
  \caption{ItemKNN}
  \label{fig:outputBiasWithAccuracyItemKNNAZ}
\end{subfigure}
\begin{subfigure}{.24\textwidth}
  \centering
  \includegraphics[scale=0.25]{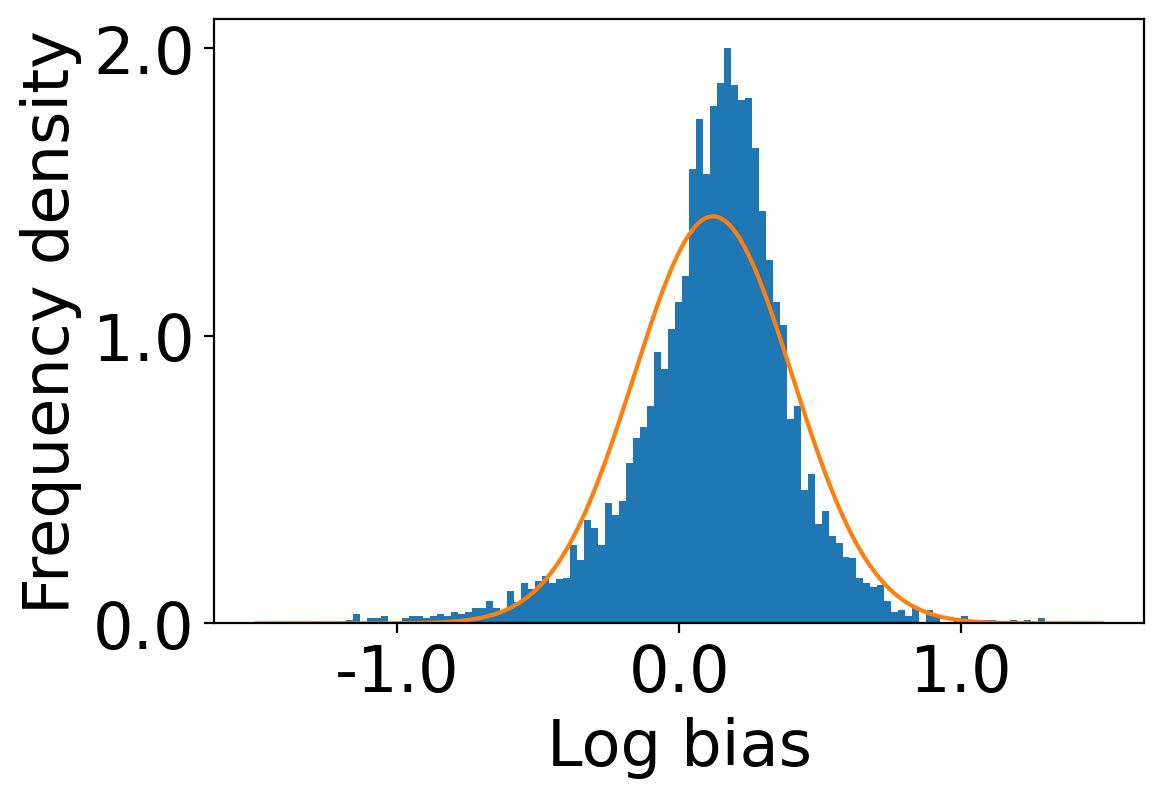}
  \caption{ALS}
  \label{fig:outputBiasWithAccuracyAlsAZ}
\end{subfigure}
\begin{subfigure}{.24\textwidth}
  \centering
  \includegraphics[scale=0.25]{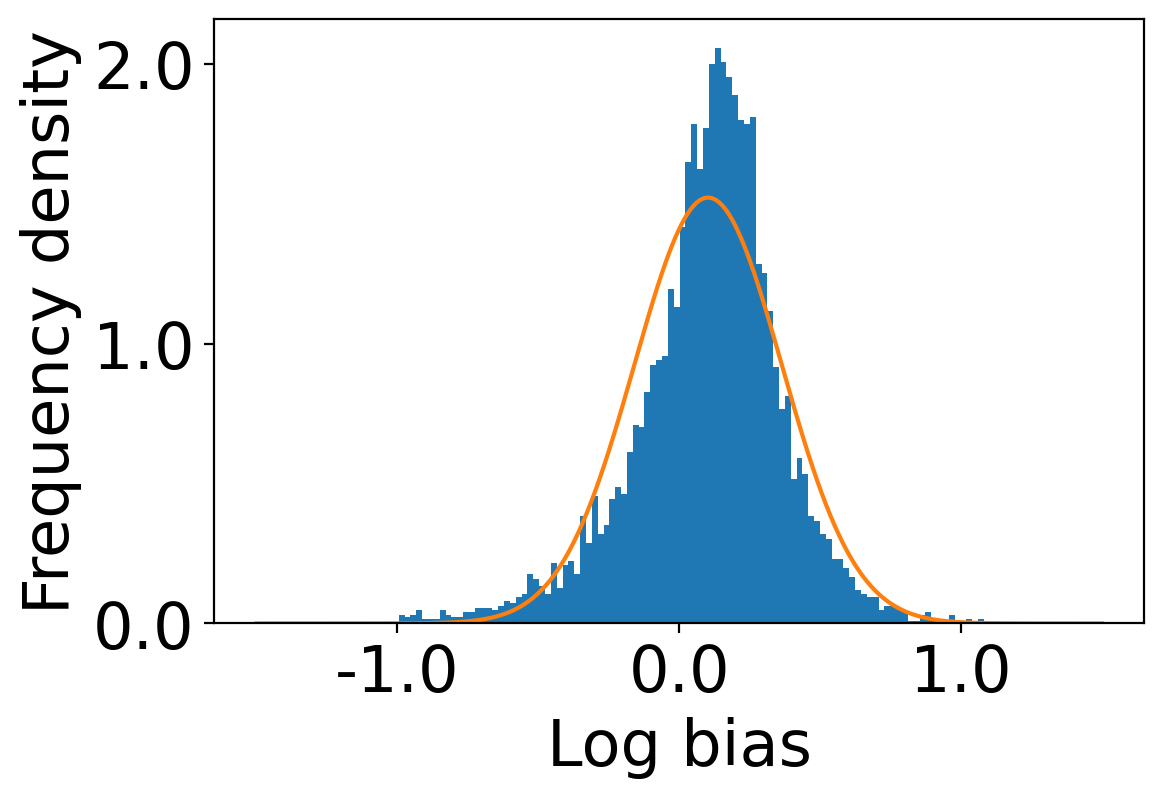}
  \caption{SVD}
  \label{fig:outputBiasWithAccuracySvdAZ}
\end{subfigure}
\caption{Output log-bias in AZ dataset with reinserting the biases}
\label{fig:outputLogBiasWithAccuracyAZ}
\end{figure*}


\begin{figure*}[]
\begin{subfigure}{.24\textwidth}
  \centering
  \includegraphics[scale=0.25]{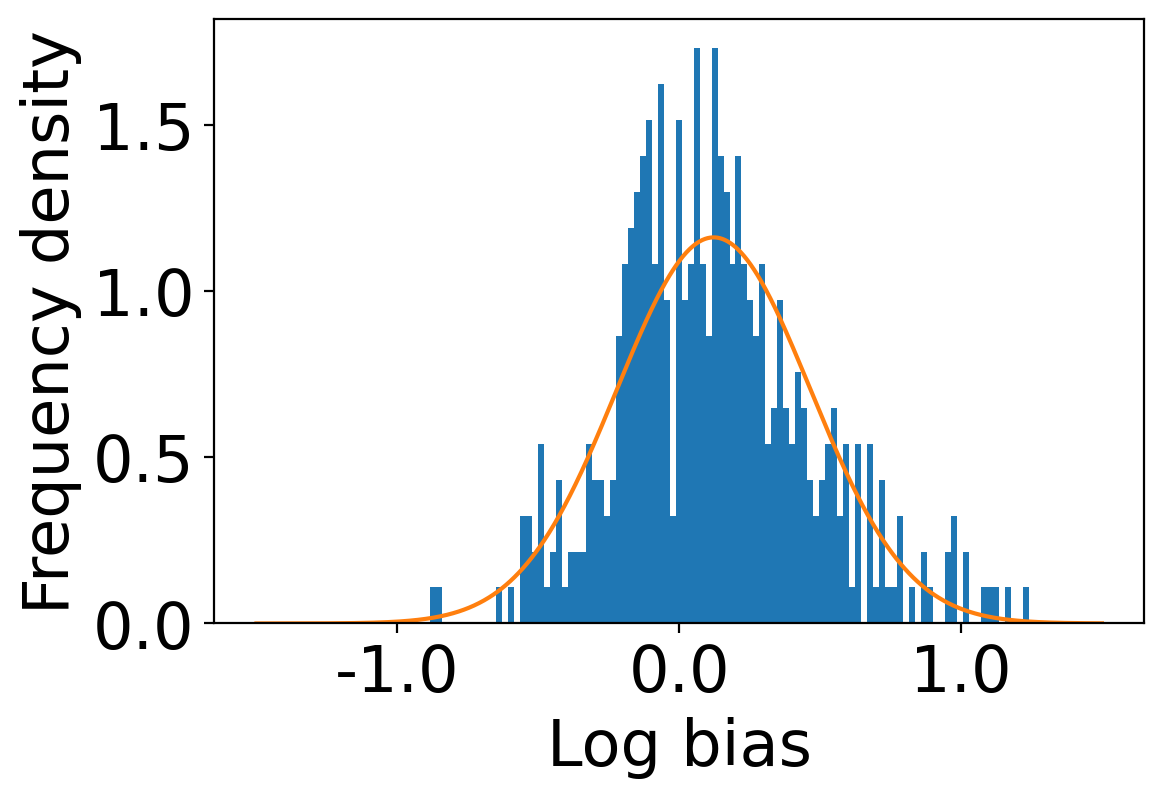}
  \caption{UserKNN}
  \label{fig:outputBiasWithoutModelUserKNNBX}
\end{subfigure}
\begin{subfigure}{.24\textwidth}
  \centering
  \includegraphics[scale=0.25]{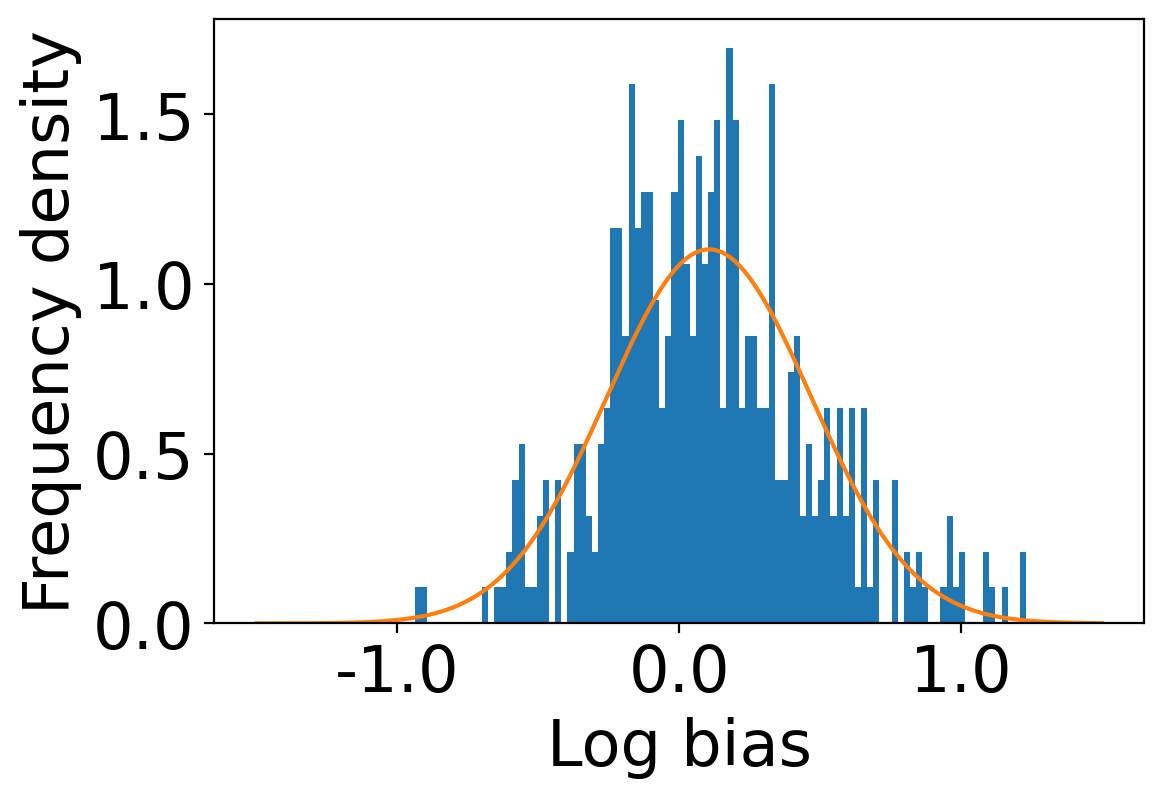}
  \caption{ItemKNN}
  \label{fig:outputBiasWithoutModelItemKNNBX}
\end{subfigure}
\begin{subfigure}{.24\textwidth}
  \centering
  \includegraphics[scale=0.25]{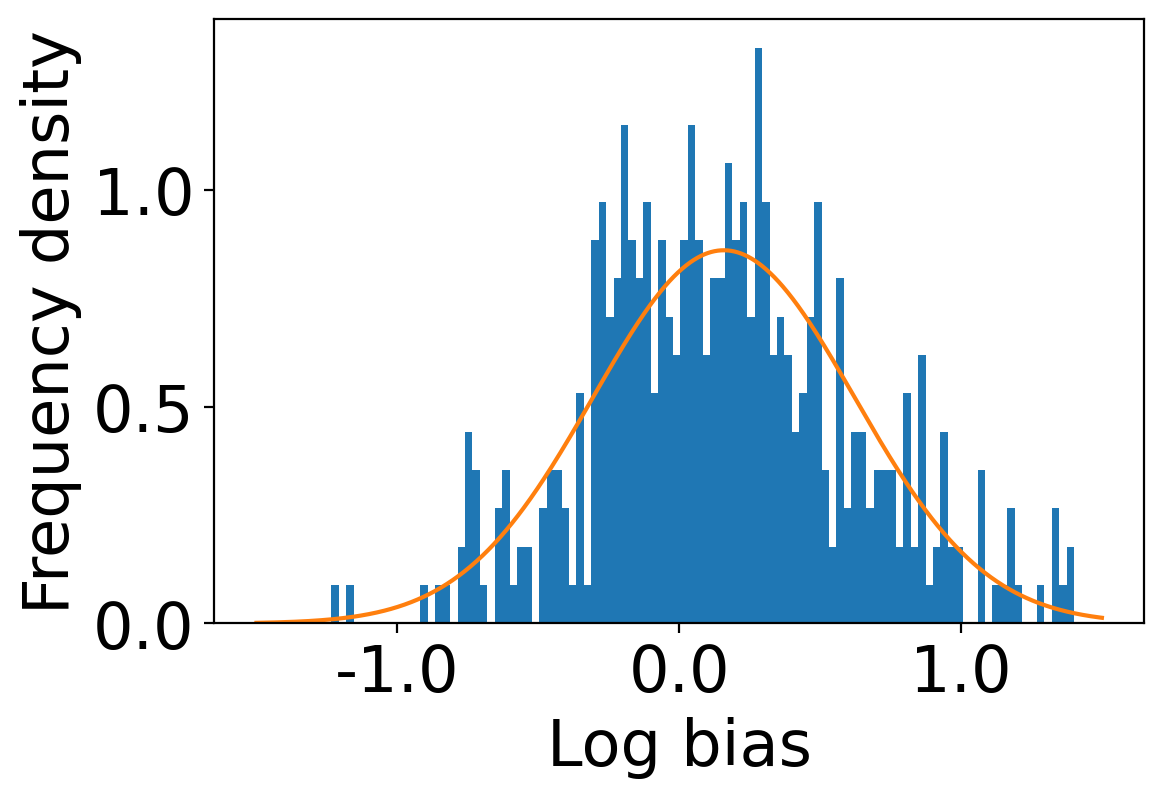}
  \caption{ALS}
  \label{fig:outputBiasWithoutModelAlsBX}
\end{subfigure}
\begin{subfigure}{.24\textwidth}
  \centering
  \includegraphics[scale=0.25]{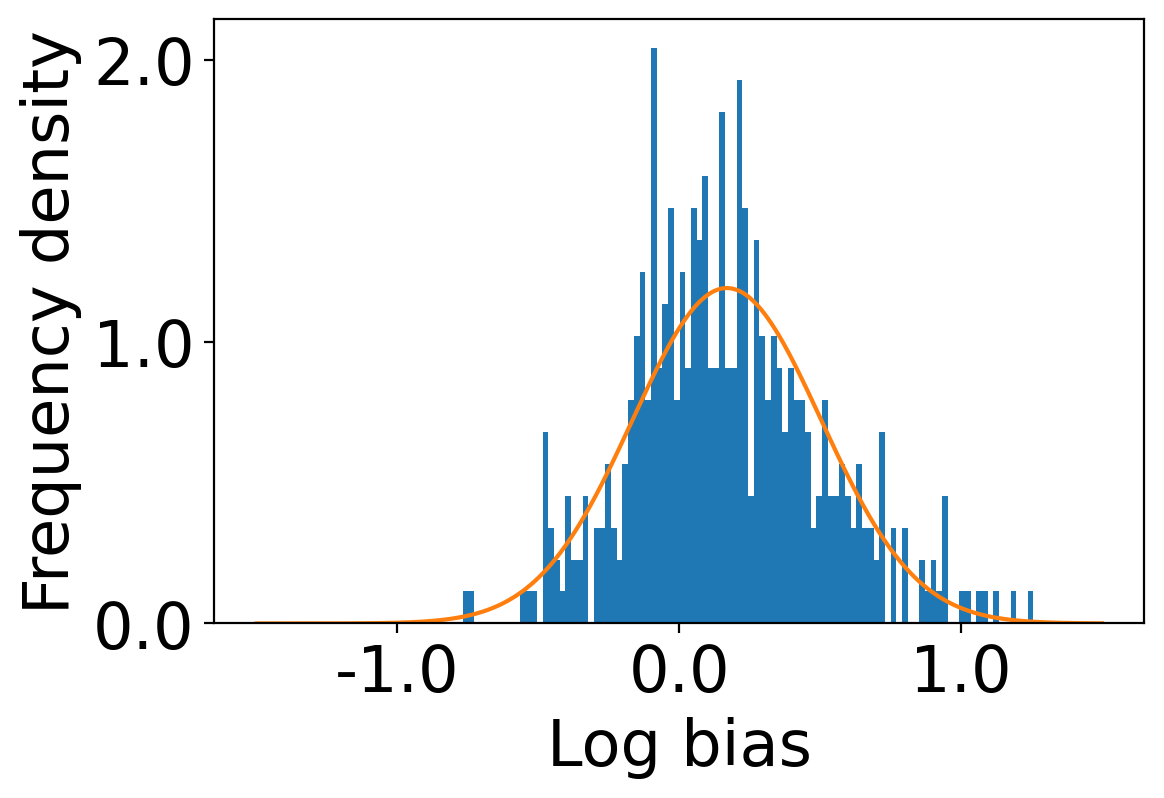}
  \caption{SVD}
  \label{fig:outputBiasWithoutModelSvdBX}
\end{subfigure}
\caption{Output log-bias in BX dataset without employing the model}
\label{fig:outputLogBiasWithoutModelBX}
\end{figure*}


\begin{figure*}[ht!]
\begin{subfigure}{.24\textwidth}
  \centering
  \includegraphics[scale=0.25]{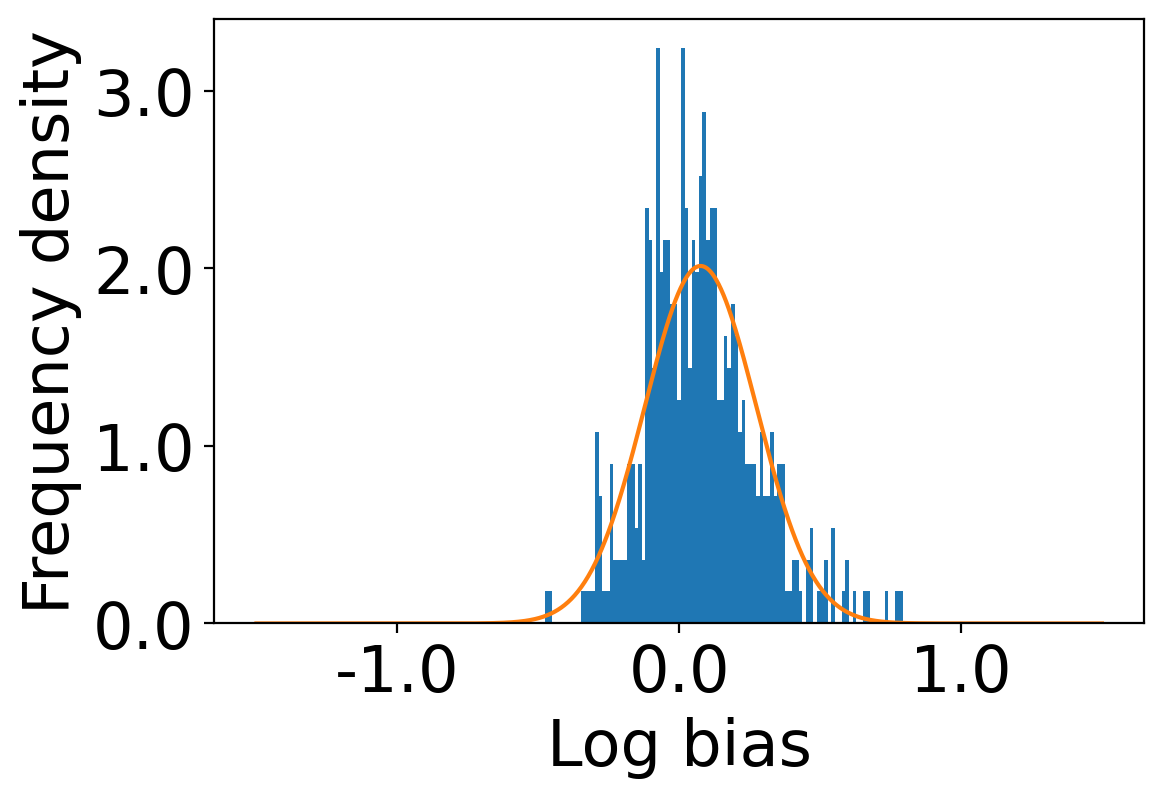}
  \caption{UserKNN}
  \label{fig:outputBiasWithoutAccuracyUserKNNBX}
\end{subfigure}
\begin{subfigure}{.24\textwidth}
  \centering
  \includegraphics[scale=0.25]{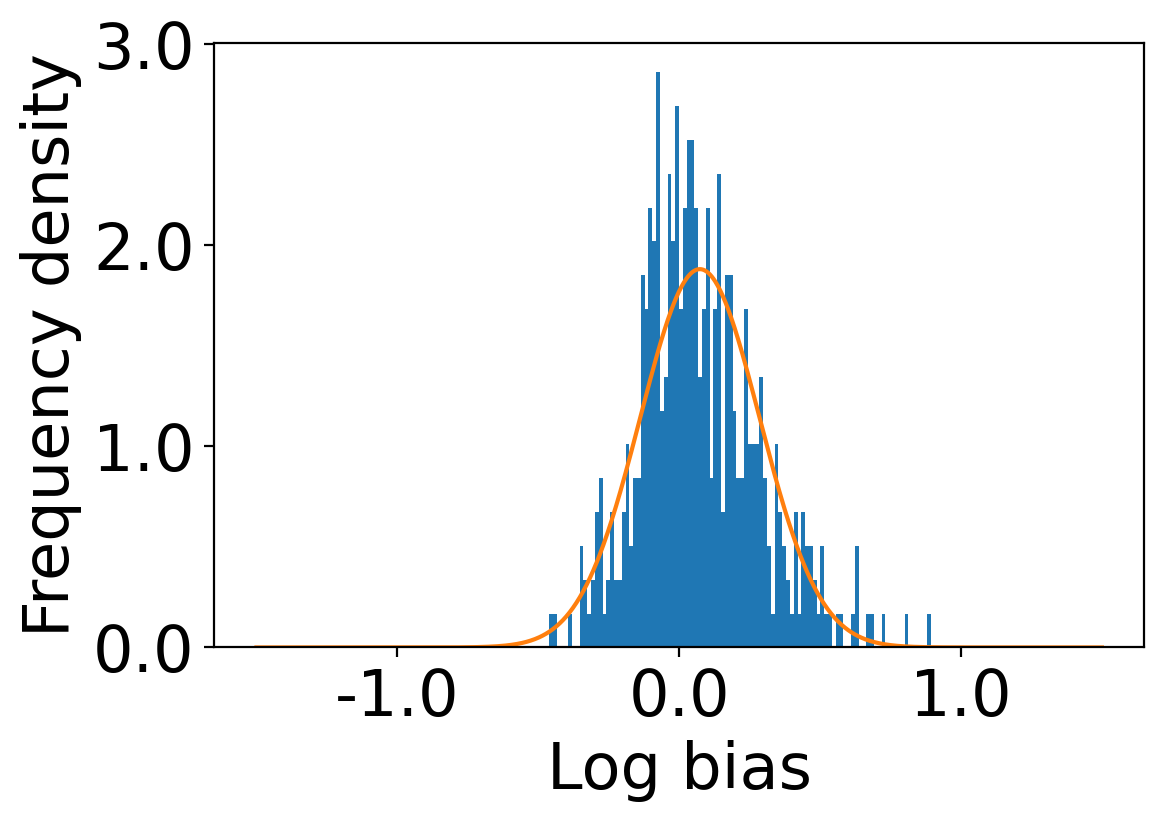}
  \caption{ItemKNN}
  \label{fig:outputBiasWithoutAccuracyItemKNNBX}
\end{subfigure}
\begin{subfigure}{.24\textwidth}
  \centering
  \includegraphics[scale=0.25]{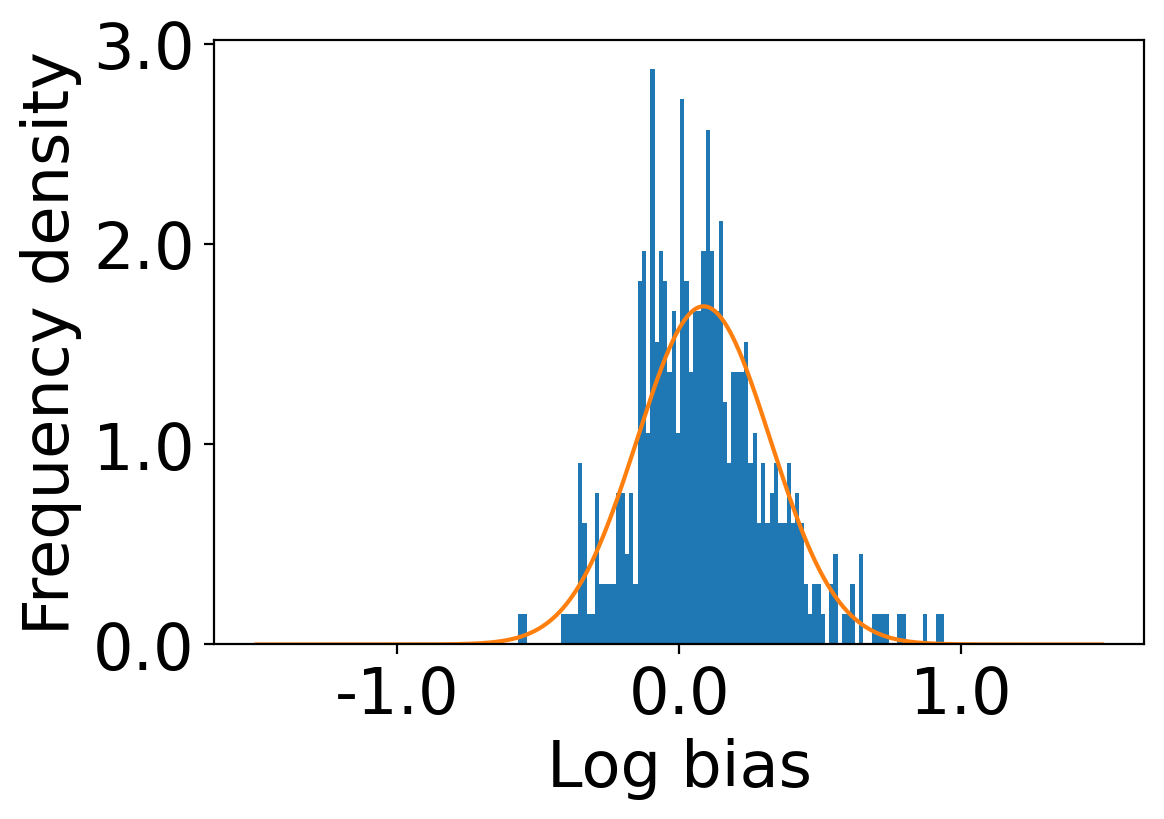}
  \caption{ALS}
  \label{fig:outputBiasWithoutAccuracyAlsBX}
\end{subfigure}
\begin{subfigure}{.24\textwidth}
  \centering
  \includegraphics[scale=0.25]{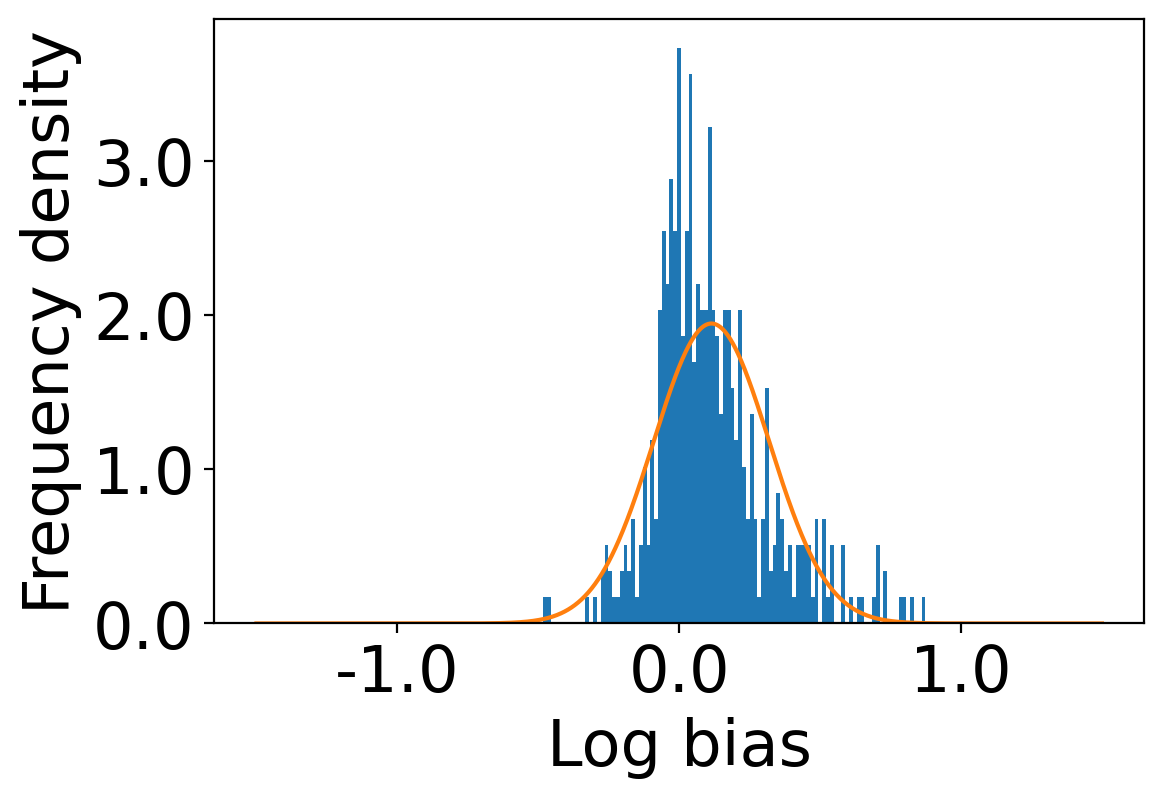}
  \caption{SVD}
  \label{fig:outputBiasWithoutAccuracySvdBX}
\end{subfigure}
\caption{Output log-bias in BX dataset with debiasing}
\label{fig:outputLogBiasWithoutAccuracyBX}
\end{figure*}

\begin{figure*}[]
\begin{subfigure}{.24\textwidth}
  \centering
  \includegraphics[scale=0.25]{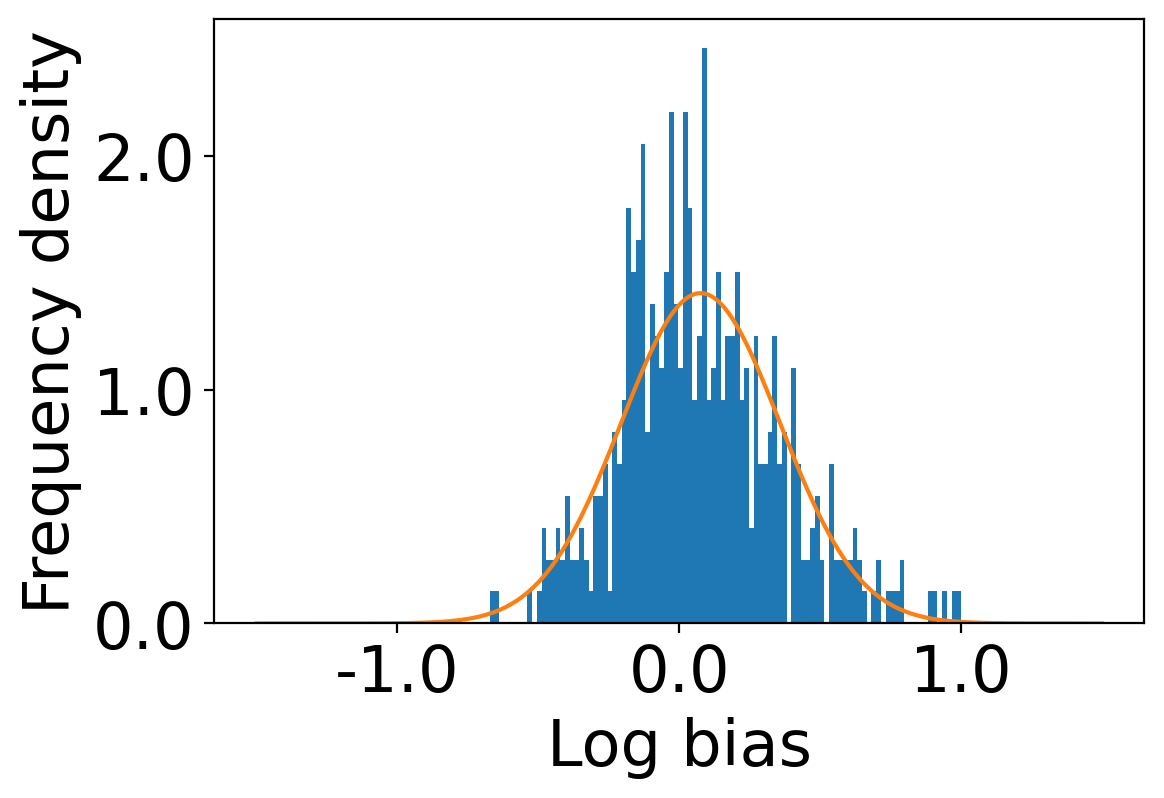}
  \caption{UserKNN}
  \label{fig:outputBiasWithAccuracyUserKNNBX}
\end{subfigure}
\begin{subfigure}{.24\textwidth}
  \centering
  \includegraphics[scale=0.25]{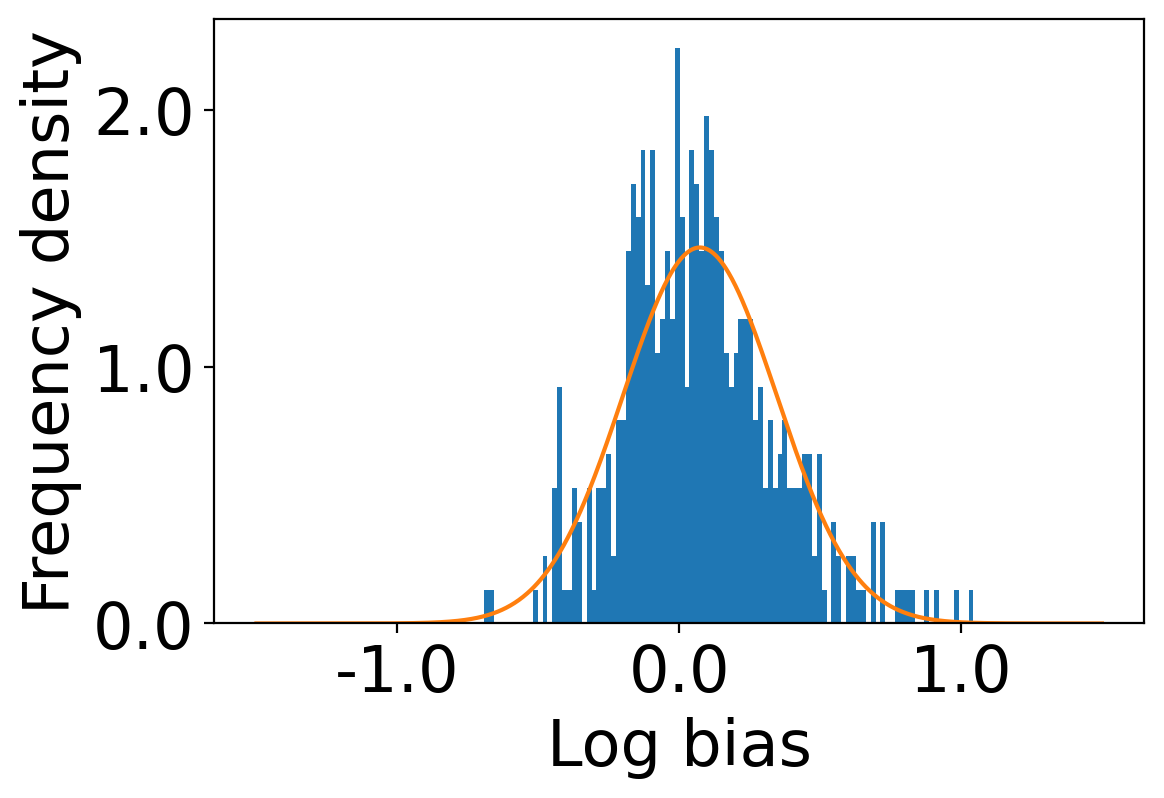}
  \caption{ItemKNN}
  \label{fig:outputBiasWithAccuracyItemKNNBX}
\end{subfigure}
\begin{subfigure}{.24\textwidth}
  \centering
  \includegraphics[scale=0.25]{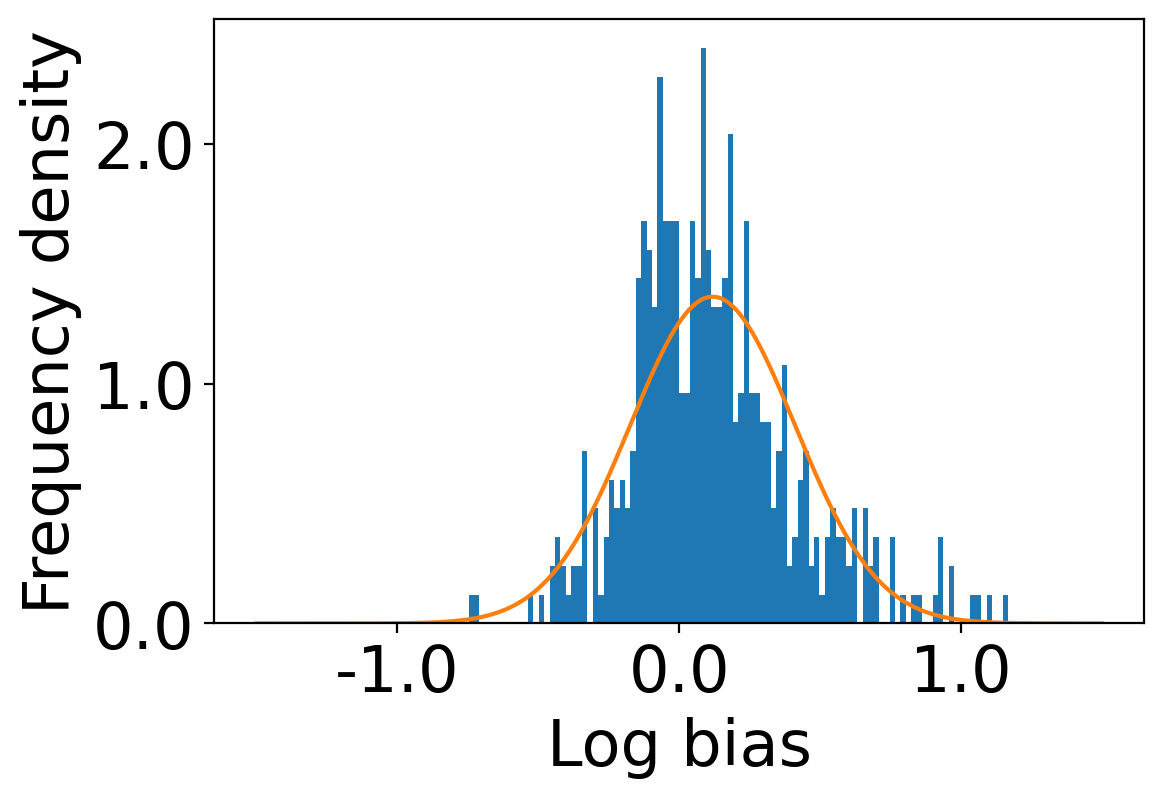}
  \caption{ALS}
  \label{fig:outputBiasWithAccuracyAlsBX}
\end{subfigure}
\begin{subfigure}{.24\textwidth}
  \centering
  \includegraphics[scale=0.25]{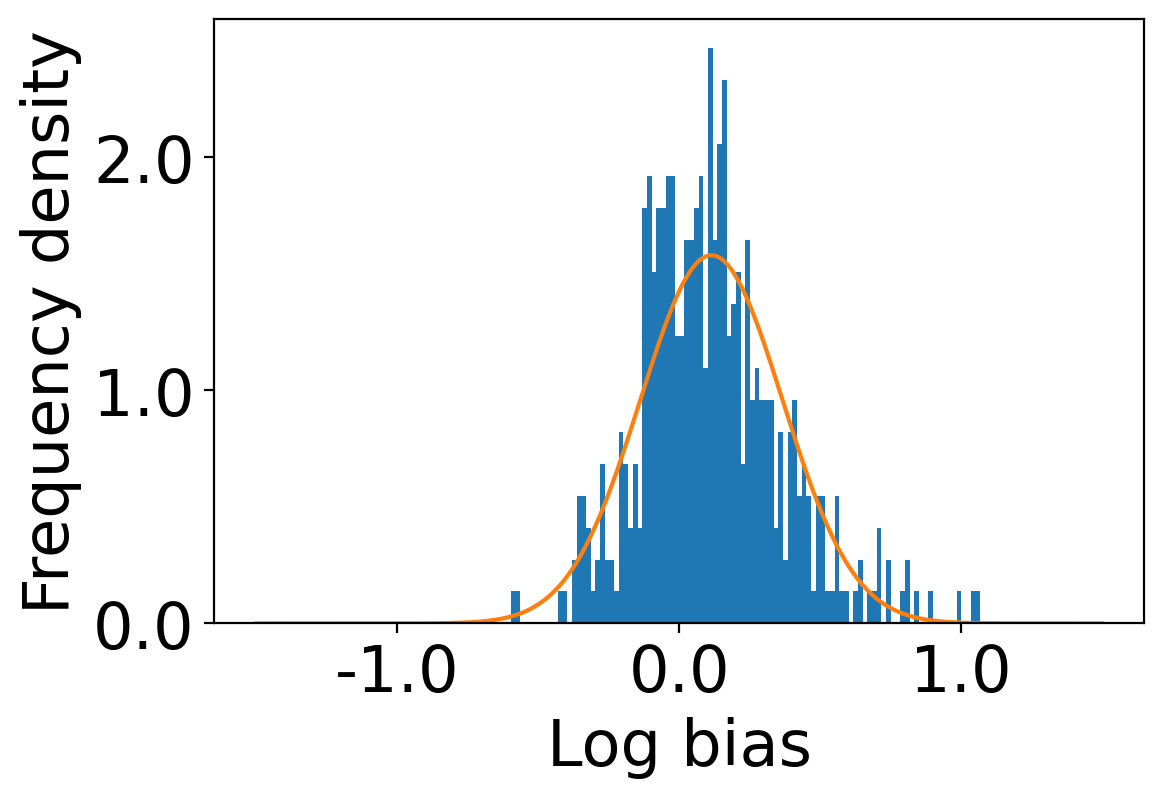}
  \caption{SVD}
  \label{fig:outputBiasWithAccuracySvdBX}
\end{subfigure}
\caption{Output log-bias in BX dataset with reinserting the biases}
\label{fig:outputLogBiasWithAccuracyBX}
\end{figure*}
\begin{table*}[ht!]
\tiny
\centering
\begin{tabular}{|c|c|c|c|c|c|c|c|c|c|c|c|}
\hline
& &\multicolumn{5}{c|}{AZ}                            & \multicolumn{5}{c|}{BX}                       \\ \hline
Case &
  Algorithm &
  \begin{tabular}[c]{@{}c@{}}Mean\\ log-bias\end{tabular} &
  RMSE &
  MAE &
  NDCG &
  \begin{tabular}[c]{@{}c@{}}Reciprocal\\ Rank\end{tabular} &
  \begin{tabular}[c]{@{}c@{}}Mean\\ log-bias\end{tabular} &
  RMSE &
  MAE &
  NDCG &
  \begin{tabular}[c]{@{}c@{}}Reciprocal\\ Rank\end{tabular} \\ \hline
\multirow{4}{*}{\begin{tabular}[c]{@{}c@{}}Recommendations \\ without model\end{tabular}} &
  UserKNN &
  0.137 &
  0.808 &
  0.693 &
  0.452 &
  0.498 &
  0.122 &
  1.580 &
  1.178 &
  0.264 &
  0.272\\ \cline{2-12} 
 & ItemKNN & 0.129 & 0.736 & 0.580  & 0.597  & 0.643 & 0.106 & 1.511 & 1.304 & 0.313  & 0.412   \\ \cline{2-12} 
 & ALS     & 0.164 & 0.873 & 0.829  & 0.281  & 0.447 & 0.158 & 1.815 & 1.642 & 0.235  & 0.370  \\ \cline{2-12} 
 & SVD     & 0.175 & 0.790 & 0.753  & 0.342  & 0.471 & 0.169 & 1.761 & 1.626 & 0.277  & 0.296 \\ \hline
\multirow{4}{*}{\begin{tabular}[c]{@{}c@{}}Recommendations without\\ preference correction\\ phase\end{tabular}} &
  UserKNN &
  0.049 &
  1.103 &
  0.921 &
  0.0224 &
  0.0278 &
  0.057 &
  2.468 &
  1.754 &
  0.0232 &
  0.0245 \\ \cline{2-12} 
  
 & ItemKNN & 0.063 & 1.076 & 0.873  & 0.0229 & 0.0204 & 0.054 & 2.463 & 2.055 & 0.0142 & 0.0271 \\ \cline{2-12} 
 & ALS     & 0.093 & 1.281 & 1.1211 & 0.0161 & 0.0394 & 0.087 & 2.752 & 2.175 & 0.0421 & 0.0736 \\ \cline{2-12} 
 & SVD     & 0.071 & 1.257 & 1.183  & 0.0138 & 0.0206 & 0.072 & 2.601 & 1.979 & 0.0261 & 0.0240\\ \hline
\multirow{4}{*}{\begin{tabular}[c]{@{}c@{}}Recommendations with\\ preference correction\\ phase\end{tabular}} &
  UserKNN &
  0.079 &
  0.871 &
  0.738 &
  0.3982 &
  0.4462 &
  0.076 &
  1.799 &
  1.298 &
  0.2099 &
  0.2317 \\ \cline{2-12} 
 & ItemKNN & 0.080 & 0.824 & 0.661  & 0.5236 & 0.6121 & 0.073 & 1.785 & 1.516 & 0.2538 & 0.3560\\ \cline{2-12} 
 & ALS     & 0.121 & 0.982 & 0.903  & 0.2391 & 0.3853 & 0.119 & 2.022 & 1.768 & 0.1626 & 0.2831\\ \cline{2-12} 
 & SVD     & 0.103 & 0.872 & 0.847  & 0.2989 & 0.4159 & 0.114 & 1.988 & 1.731 & 0.2258 & 0.2481 \\ \hline
\end{tabular}
\caption{Summary of Results}
\label{tab:azResults}
\end{table*}

\begin{figure}[]
\begin{subfigure}{.45\textwidth}
  \centering
  \includegraphics[scale=0.4]{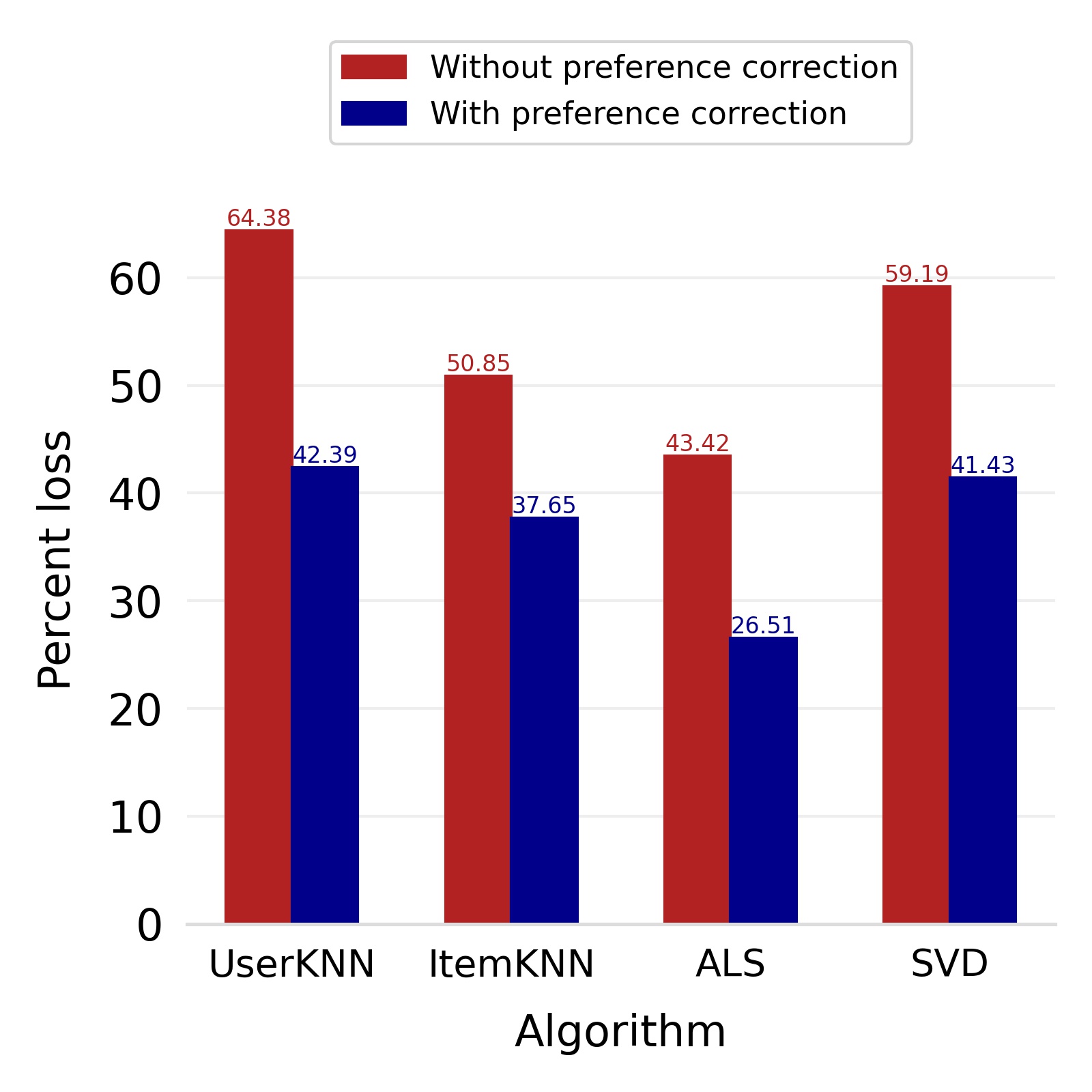}  
  \caption{AZ dataset}
  \label{fig:biasReductionAZ}
\end{subfigure}
\begin{subfigure}{.45\textwidth}
  \centering
  \includegraphics[scale=0.4]{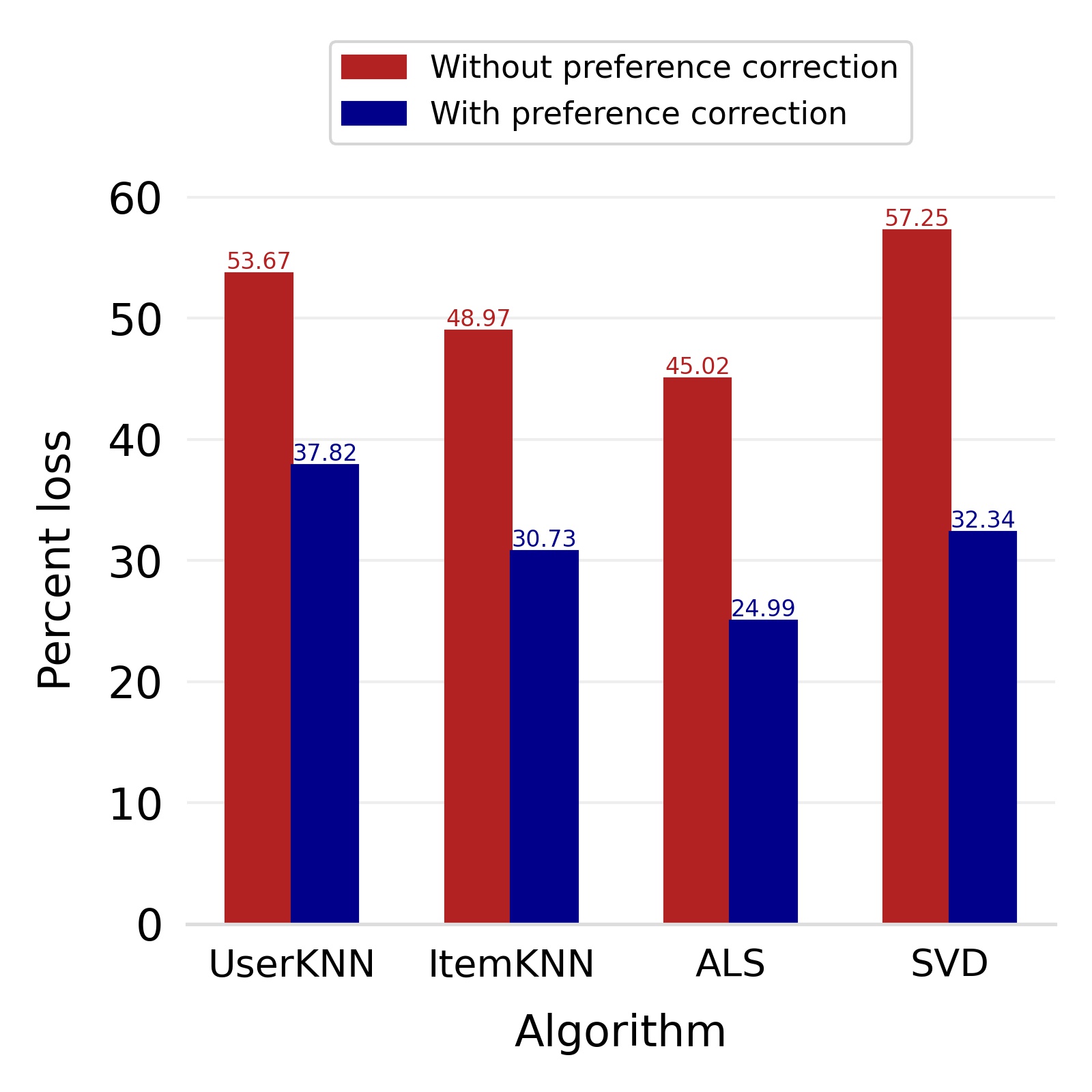}  
  \caption{BX dataset}
  \label{fig:biasReductionBX}
\end{subfigure}
\caption{Bias reduction}
\label{fig:biasReduction}
\end{figure}

\begin{figure}[ht!]
\begin{subfigure}{.45\textwidth}
  \centering
  \includegraphics[scale=0.4]{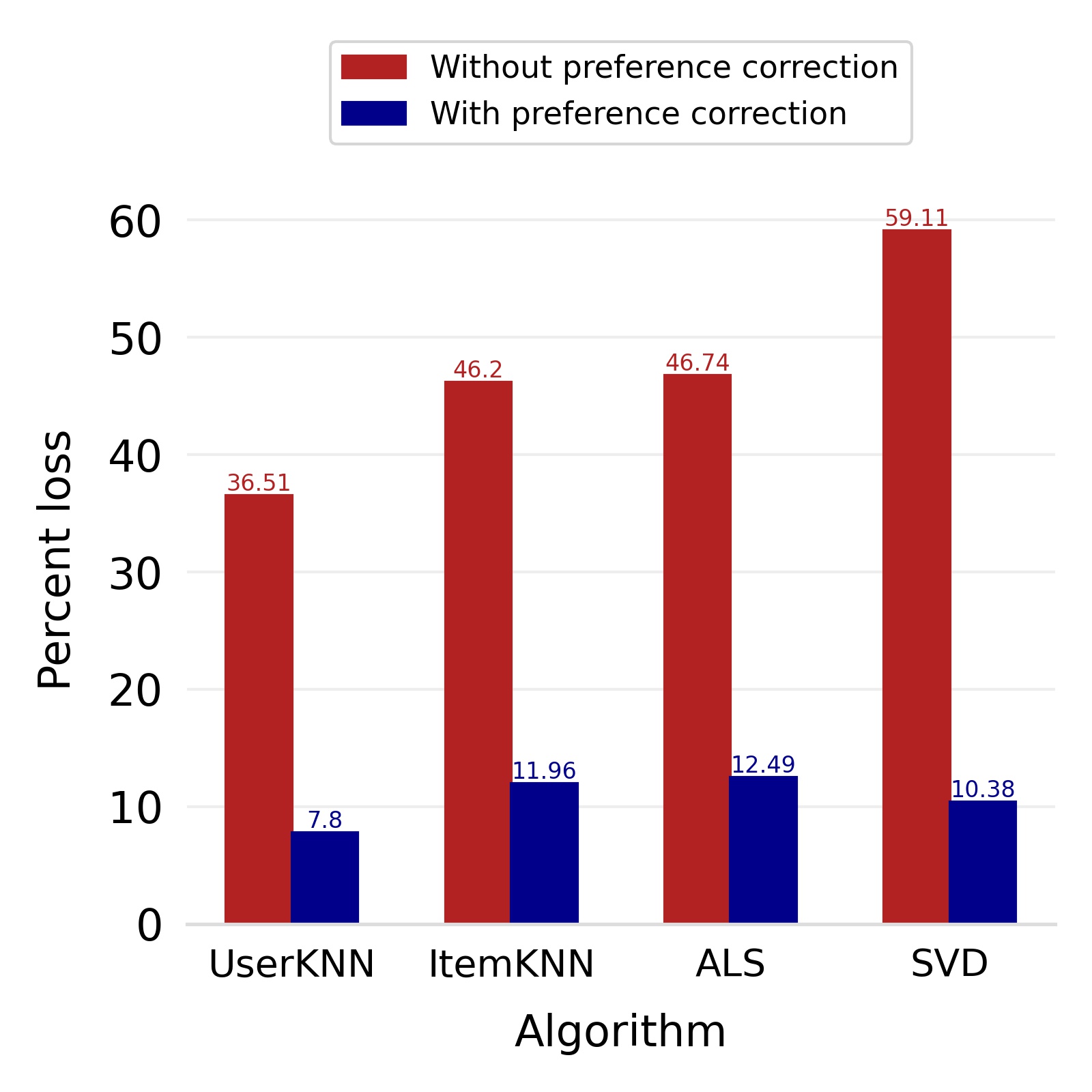}  
  \caption{In terms of RMSE}
  \label{fig:rmseAZ}
\end{subfigure}
\begin{subfigure}{.45\textwidth}
  \centering
  \includegraphics[scale=0.4]{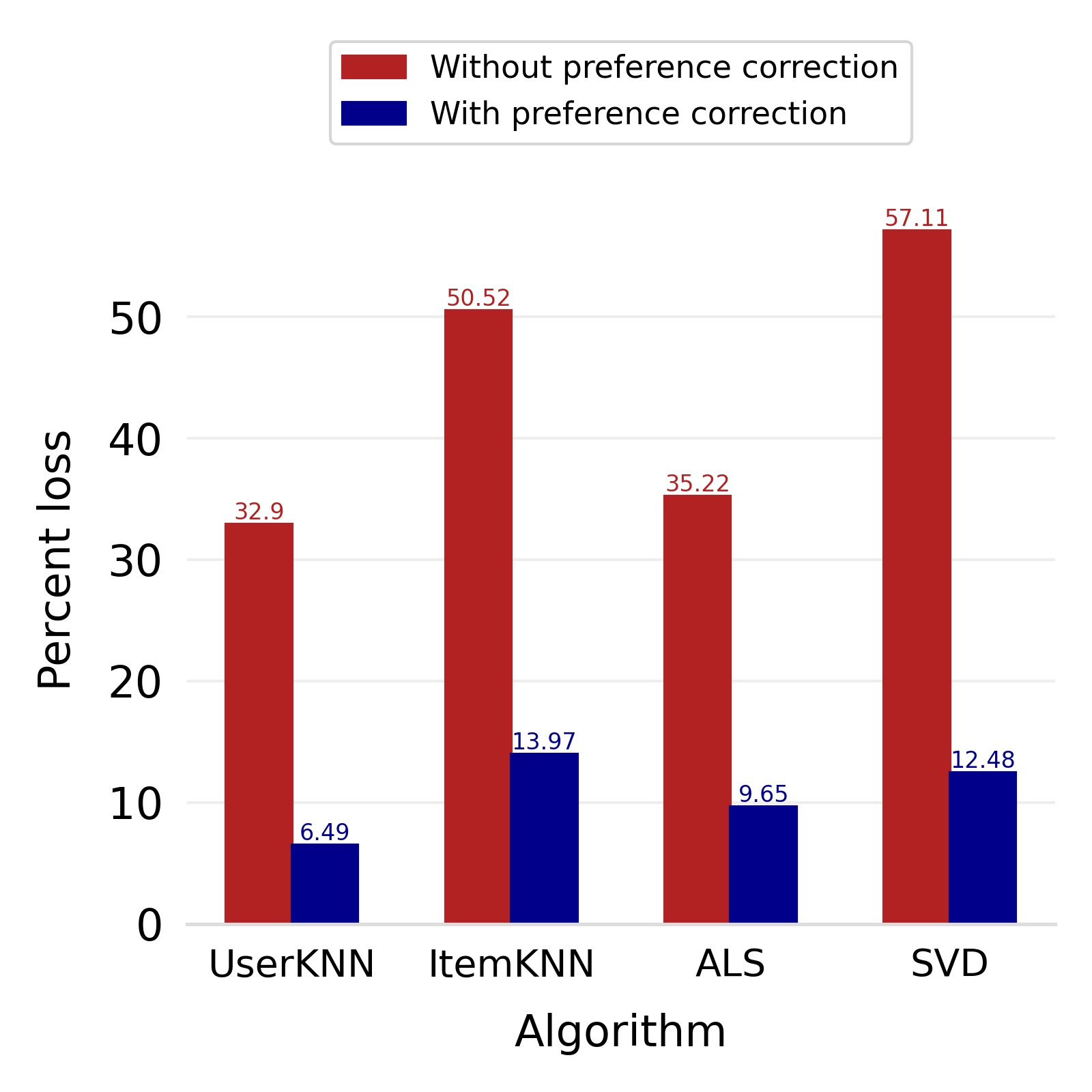}  
  \caption{In terms of MAE}
  \label{fig:maeAZ}
\end{subfigure}
\caption{Accuracy loss for AZ dataset}
\label{fig:accuracyLossAZ}
\end{figure}
\begin{figure}[ht!]
\begin{subfigure}{.45\textwidth}
  \centering
  \includegraphics[scale=0.4]{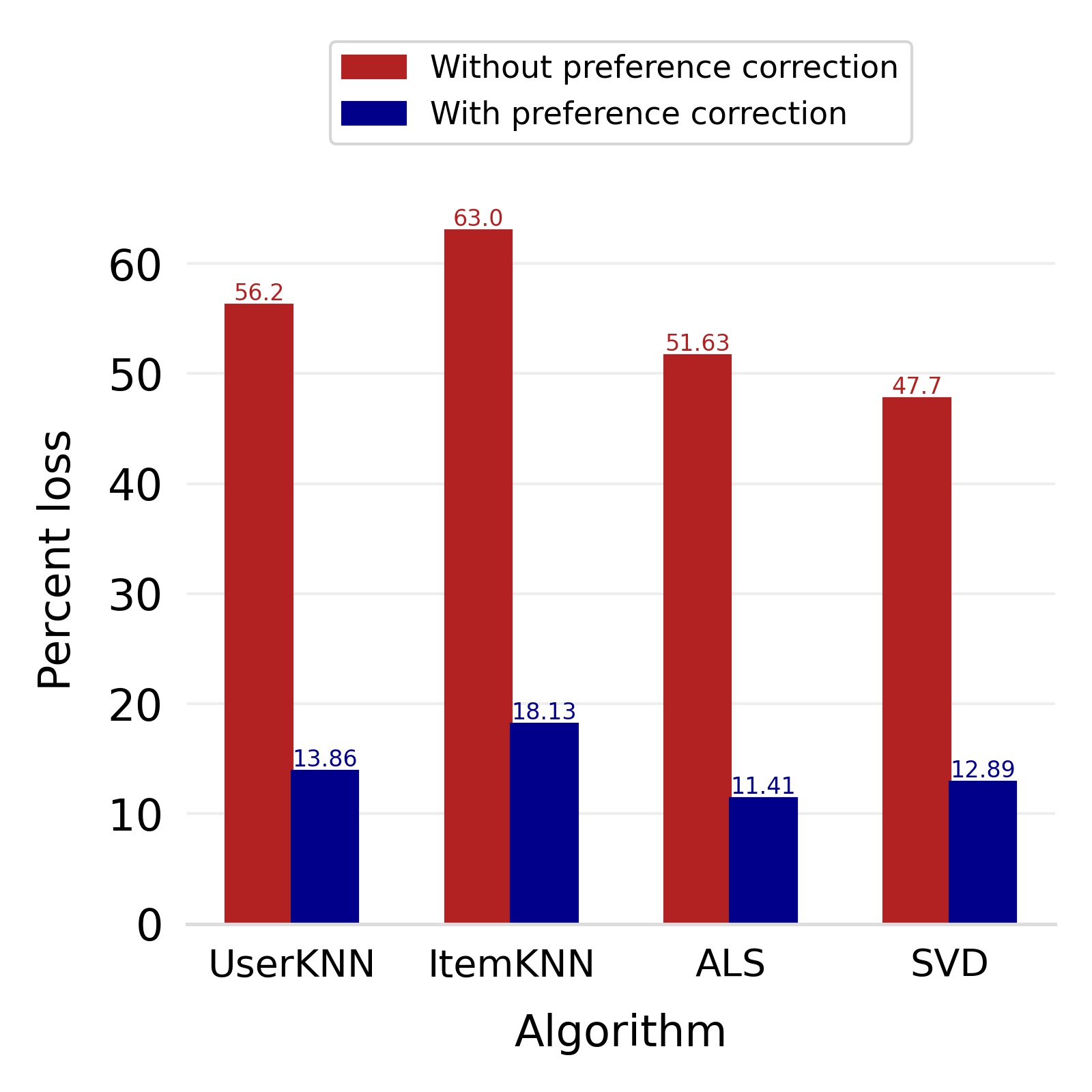}  
  \caption{In terms of RMSE}
  \label{fig:rmseBX}
\end{subfigure}
\begin{subfigure}{.45\textwidth}
  \centering
  \includegraphics[scale=0.4]{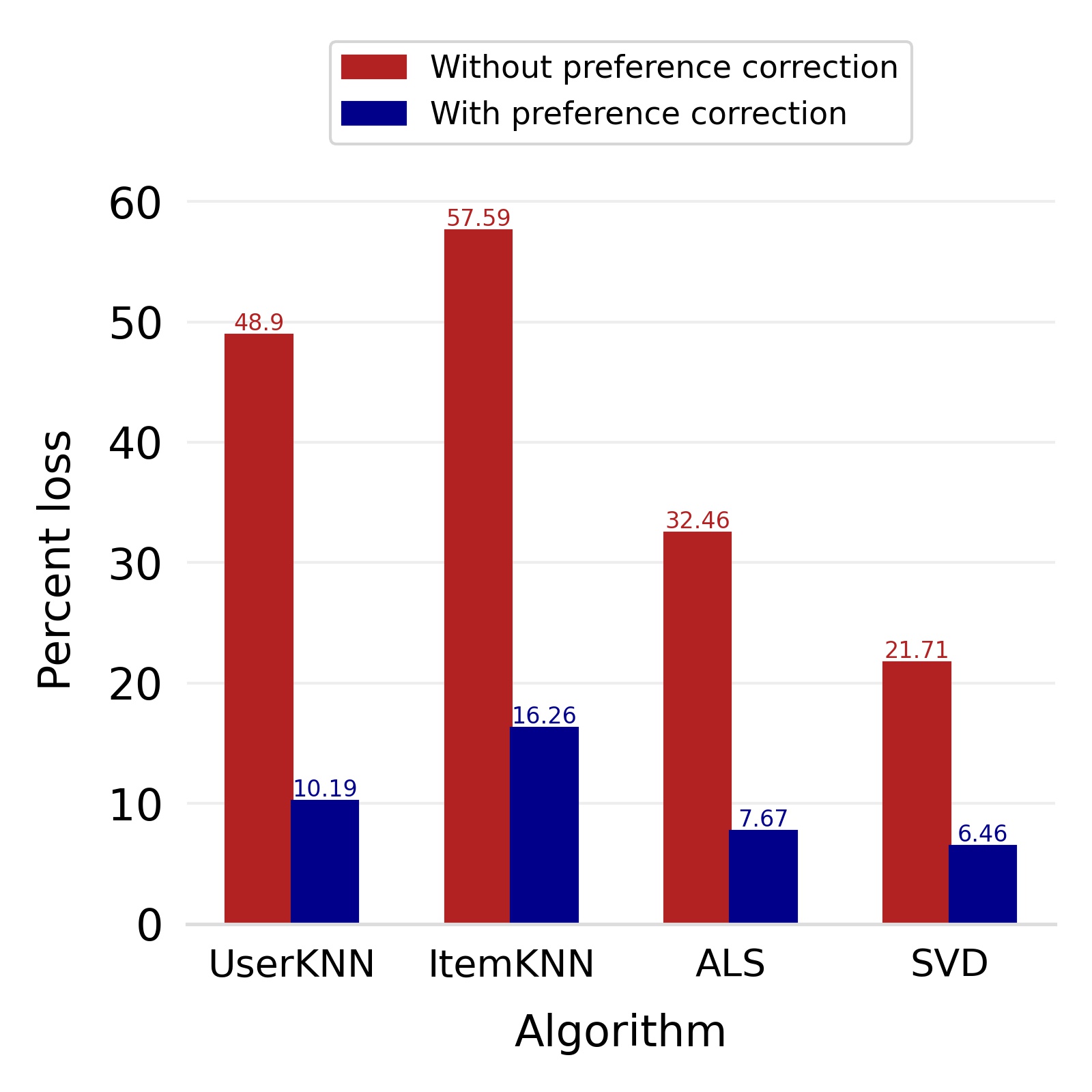}  
  \caption{In terms of MAE}
  \label{fig:maeBX}
\end{subfigure}
\caption{Accuracy loss for BX dataset}
\label{fig:accuracyLossBX}
\end{figure}


\begin{figure}[ht!]
\begin{subfigure}{.45\textwidth}
  \centering
  \includegraphics[scale=0.4]{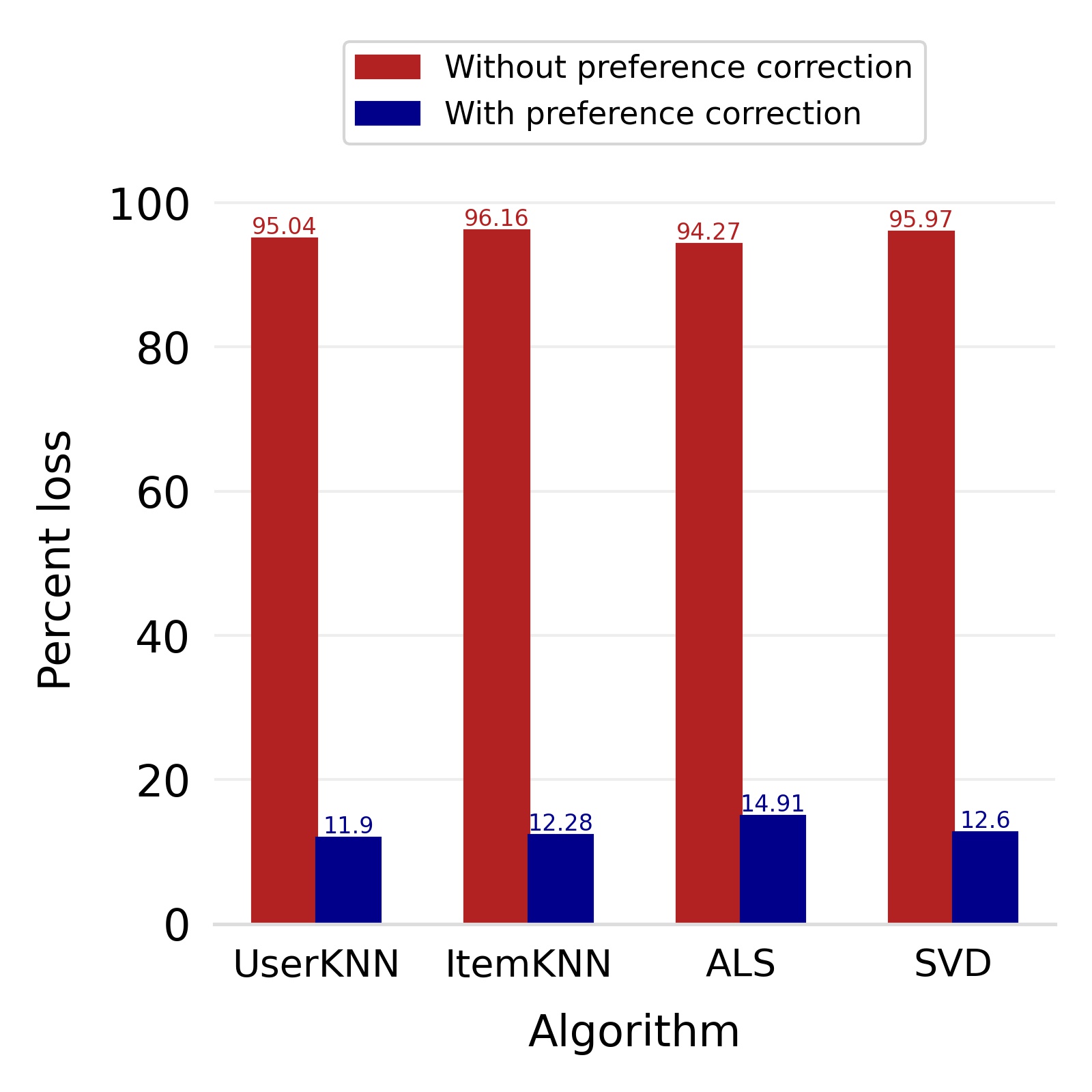}  
  \caption{In terms of NDCG}
  \label{fig:ndcgAZ}
\end{subfigure}
\begin{subfigure}{.45\textwidth}
  \centering
  \includegraphics[scale=0.4]{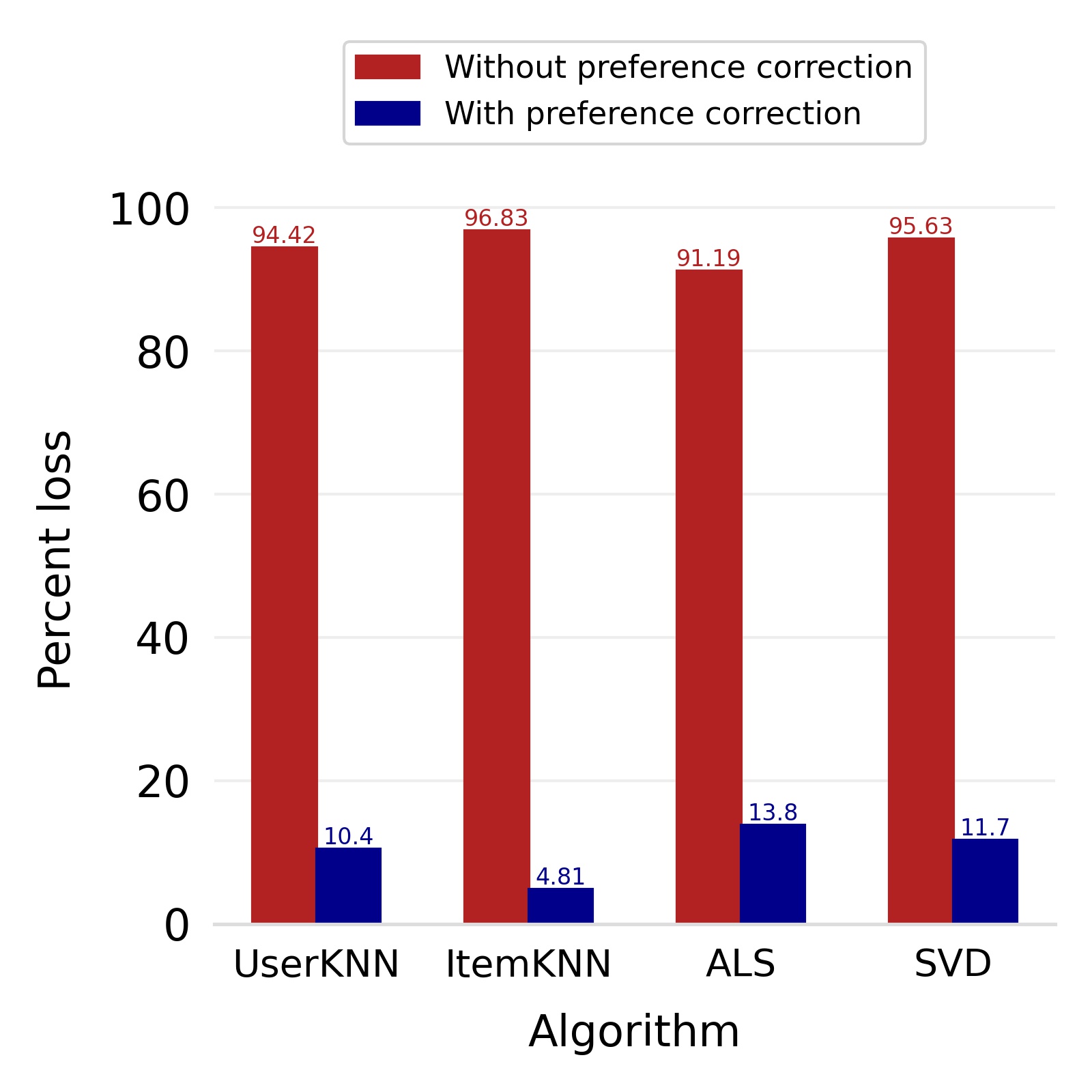}  
  \caption{In terms of Reciprocal Rank}
  \label{fig:recipAZ}
\end{subfigure}
\caption{Ranking relevancy loss for AZ dataset}
\label{fig:rankingRelevancyLossAZ}
\end{figure}




\begin{figure}[ht!]
\begin{subfigure}{.45\textwidth}
  \centering
  \includegraphics[scale=0.4]{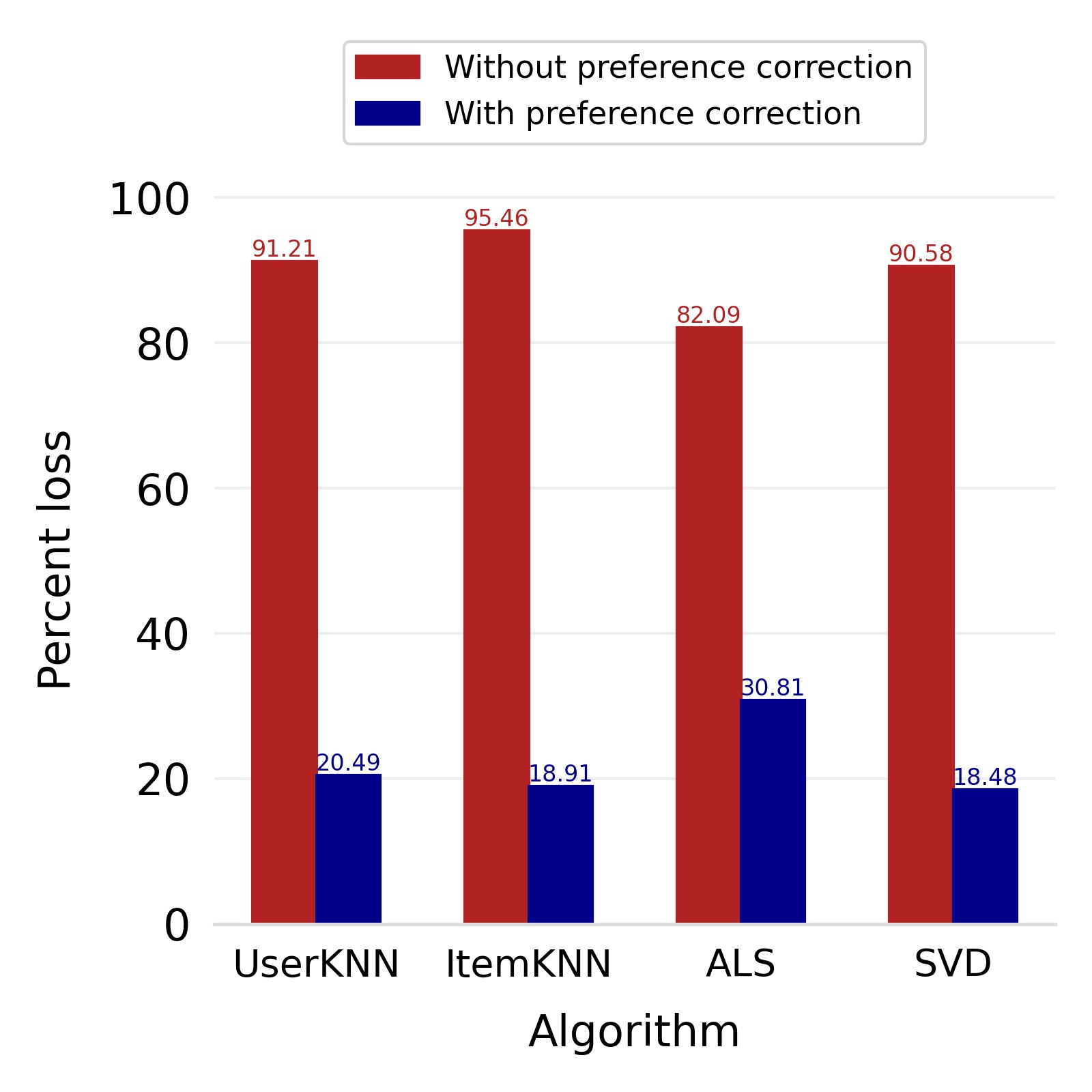}  
  \caption{In terms of NDCG}
  \label{fig:ndcgBX}
\end{subfigure}
\begin{subfigure}{.45\textwidth}
  \centering
  \includegraphics[scale=0.4]{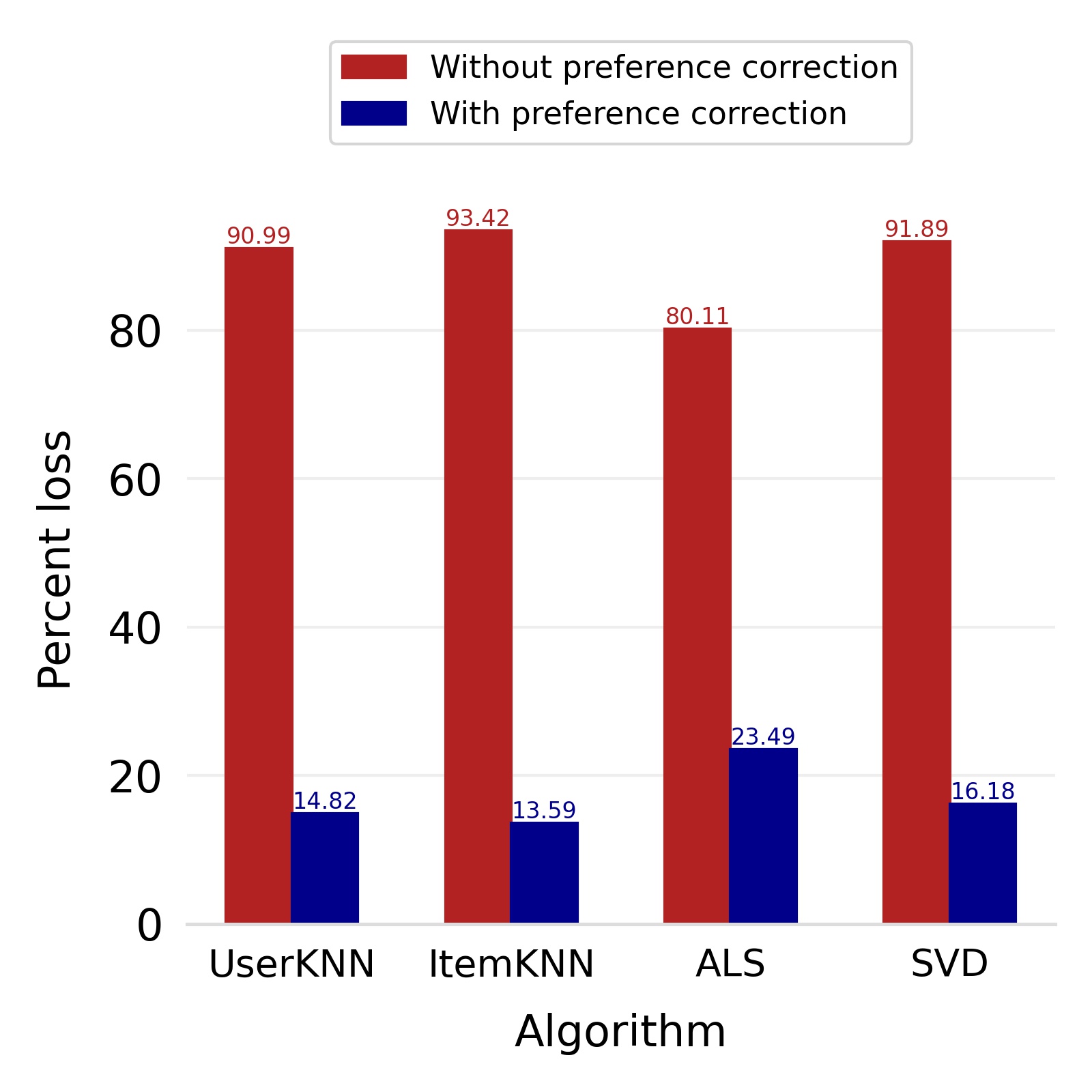}  
  \caption{In terms of Reciprocal Rank}
  \label{fig:recipBX}
\end{subfigure}
\caption{Ranking relevancy loss for BX dataset}
\label{fig:rankingRelevancyLossBX}
\end{figure}


We next conduct significance testing to validate the log-bias reduction. Table \ref{tab:significanceTestResults} shows the p-values obtained from left-tail significance tests on the log-bias of the recommendations made for the users in the sample. We can see from the p-value for the Amazon datasets that the bias reduction is significant. For the Book-Crossing dataset, the significance of the bias reduction is less pronounced. One of the prominent reasons for this is that the test sample size for the Book-Crossing dataset was relatively small due to the small number of users in the dataset. In essence, the utility of the recommender system is maintained while reducing the log-bias tendency in the recommendations.

\begin{table*}[]
\centering
\small
\begin{tabular}{|c|c|c|c|c|c|c|c|c|c|c|}
\hline
          & \multicolumn{5}{c|}{AZ}                            & \multicolumn{5}{c|}{BX}                       \\ \hline
Algorithm & $\bar{x}$ & $\mu$ & $\sigma$ & $z$    & $p$        & $\bar{x}$ & $\mu$ & $\sigma$ & $z$    & $p$   \\ \hline
UserKNN   & 0.079     & 0.137 & 0.307    & -17.90 & $<10^{-5}$ & 0.076     & 0.122 & 0.343    & -1.164 & 0.122 \\ \hline
ItemKNN   & 0.080     & 0.129 & 0.381    & -12.06 & $<10^{-5}$ & 0.073     & 0.106 & 0.362    & -0.780 & 0.218 \\ \hline
ALS       & 0.121     & 0.164 & 0.394    & -10.46 & $<10^{-5}$ & 0.119     & 0.158 & 0.464    & -0.738 & 0.230 \\ \hline
SVD       & 0.103     & 0.175 & 0.354    & -19.27 & $<10^{-5}$ & 0.114     & 0.169 & 0.335    & -1.413 & 0.079 \\ \hline
\end{tabular}
\caption{Significance test results for bias reduction}
\label{tab:significanceTestResults}
\end{table*}

We further observe that the bias reduction is more in the case of UserKNN based recommendations than the ItemKNN based recommendations. This observation can be attributed to the fact that our model addresses the bias originating from the distortion in ratings from the users' side. It compares the ratings of an item given by a particular user with the appropriately scaled average of ratings given by other users to that item in the dataset. It, therefore, resonates with the UserKNN algorithm, which predicts the ratings of an item for a particular user based on the ratings of that item for his or her peers. The ItemKNN algorithm, on the other hand, predicts the ratings of an item for a particular user based on the ratings given to similar items by that user. The model does not sit squarely with ItemKNN. Thus the bias reduction in UserKNN is more as compared to that in the case of ItemKNN. We further observe that the bias reduction is more in the case of the AZ dataset as compared to the BX dataset. This observation can be attributed to the AZ dataset having a higher input mean log-bias tendency. Further, the AZ dataset has a significantly larger number of users and items which leads to a more accurate estimation of user bias scores and, therefore, more effective bias mitigation.

We observe that accuracy and ranking relevancy loss is, in general, higher for ItemKNN as compared to UserKNN. This is due to the fact that the model quantifies the bias of users by comparing the ratings given by them to particular items with a scaled average of ratings given by their peers to those items. This resonates with the UserKNN algorithm, which predicts user ratings for particular items based on the ratings of similar users. Thus the model is better oriented towards the UserKNN algorithm, giving better accuracy and bias reduction in its case. In the case of matrix factorization algorithms, the accuracy and ranking relevancy losses are relatively comparable. It is not clear which one of the two algorithms is more coherent with the model.

We further observe that accuracy loss on BX dataset is higher than that of AZ dataset. This observation can be attributed to the fact that the user and item base of the AZ dataset is higher as compared to the BX dataset. Thus, the bias score estimates are more accurate, which provides more accurate predictions of the item scores for the users when reinserted into the recommendations.



\section{Conclusion and Future Work}
We proposed a model to quantify and mitigate the bias in the explicit feedback given by the users to different items. We theoretically showed that the debiased ratings produced by our model are unbiased estimators of the true preference of the users for the books. With the help of comprehensive experiments on two publically available book datasets, we show a significant reduction in the bias (almost 40\%) with just 10\% decrease in accuracy using the UserKNN algorithm. Similar trends were observed for other algorithms such as ItemKNN, ALS, and SVD. Our model is independent of these algorithms' choices and can be applied with any recommendation algorithm. We used book recommender system because we were able to generate the gender information from publicly available APIs. Our model is not restricted to book recommender system as long as protected attribute information about the items is known. We leave extension of the model to missing protected attribute as an interesting future work. It will be an interesting direction to see if the ideas from fair classification literature with missing protected attributes \cite{coston2019fair} can be leveraged. We further did not address the bias originating from fewer ratings for a female-authored book than a male-authored one. 
We leave extending the model to the bias originating from lesser number of ratings and extensively studying the model for other recommender systems as the future directions.
\newpage 
\bibliographystyle{ACM-Reference-Format}
\bibliography{sample-base}


\begin{thebibliography}{36}


\ifx \showCODEN    \undefined \def \showCODEN     #1{\unskip}     \fi
\ifx \showDOI      \undefined \def \showDOI       #1{#1}\fi
\ifx \showISBNx    \undefined \def \showISBNx     #1{\unskip}     \fi
\ifx \showISBNxiii \undefined \def \showISBNxiii  #1{\unskip}     \fi
\ifx \showISSN     \undefined \def \showISSN      #1{\unskip}     \fi
\ifx \showLCCN     \undefined \def \showLCCN      #1{\unskip}     \fi
\ifx \shownote     \undefined \def \shownote      #1{#1}          \fi
\ifx \showarticletitle \undefined \def \showarticletitle #1{#1}   \fi
\ifx \showURL      \undefined \def \showURL       {\relax}        \fi
\providecommand\bibfield[2]{#2}
\providecommand\bibinfo[2]{#2}
\providecommand\natexlab[1]{#1}
\providecommand\showeprint[2][]{arXiv:#2}

\bibitem[\protect\citeauthoryear{Amatriain, Jaimes, Oliver, and
  Pujol}{Amatriain et~al\mbox{.}}{2011}]%
        {dataminingCriticisation}
\bibfield{author}{\bibinfo{person}{Xavier Amatriain},
  \bibinfo{person}{Alejandro Jaimes}, \bibinfo{person}{Nuria Oliver}, {and}
  \bibinfo{person}{Josep Pujol}.} \bibinfo{year}{2011}\natexlab{}.
\newblock \bibinfo{booktitle}{\emph{Data Mining Methods for Recommender
  Systems}}.
\newblock \bibinfo{pages}{39--71}.
\newblock
\urldef\tempurl%
\url{https://doi.org/10.1007/978-0-387-85820-3_2}
\showDOI{\tempurl}


\bibitem[\protect\citeauthoryear{APIs}{APIs}{2 24}]%
        {googleAPI}
\bibfield{author}{\bibinfo{person}{Google~Books APIs}.}
  \bibinfo{year}{Accessed: 2021-02-24}\natexlab{}.
\newblock
\newblock
\urldef\tempurl%
\url{https://developers.google.com/books}
\showURL{%
\tempurl}


\bibitem[\protect\citeauthoryear{Boratto, Fenu, and Marras}{Boratto
  et~al\mbox{.}}{2019}]%
        {fairrecsys}
\bibfield{author}{\bibinfo{person}{Ludovico Boratto}, \bibinfo{person}{Gianni
  Fenu}, {and} \bibinfo{person}{Mirko Marras}.}
  \bibinfo{year}{2019}\natexlab{}.
\newblock \bibinfo{booktitle}{\emph{The Effect of Algorithmic Bias on
  Recommender Systems for Massive Open Online Courses}}.
\newblock \bibinfo{pages}{457--472}.
\newblock
\showISBNx{978-3-030-15711-1}
\urldef\tempurl%
\url{https://doi.org/10.1007/978-3-030-15712-8_30}
\showDOI{\tempurl}


\bibitem[\protect\citeauthoryear{Burke}{Burke}{2017}]%
        {multisidedfairness}
\bibfield{author}{\bibinfo{person}{Robin Burke}.}
  \bibinfo{year}{2017}\natexlab{}.
\newblock \showarticletitle{Multisided Fairness for Recommendation}.
\newblock  (\bibinfo{date}{07} \bibinfo{year}{2017}).
\newblock


\bibitem[\protect\citeauthoryear{Coston, Ramamurthy, Wei, Varshney, Speakman,
  Mustahsan, and Chakraborty}{Coston et~al\mbox{.}}{2019}]%
        {coston2019fair}
\bibfield{author}{\bibinfo{person}{Amanda Coston},
  \bibinfo{person}{Karthikeyan~Natesan Ramamurthy}, \bibinfo{person}{Dennis
  Wei}, \bibinfo{person}{Kush~R Varshney}, \bibinfo{person}{Skyler Speakman},
  \bibinfo{person}{Zairah Mustahsan}, {and} \bibinfo{person}{Supriyo
  Chakraborty}.} \bibinfo{year}{2019}\natexlab{}.
\newblock \showarticletitle{Fair transfer learning with missing protected
  attributes}. In \bibinfo{booktitle}{\emph{Proceedings of the 2019 AAAI/ACM
  Conference on AI, Ethics, and Society}}. \bibinfo{pages}{91--98}.
\newblock


\bibitem[\protect\citeauthoryear{Dwork, Hardt, Pitassi, Reingold, and
  Zemel}{Dwork et~al\mbox{.}}{2011}]%
        {datamining5}
\bibfield{author}{\bibinfo{person}{Cynthia Dwork}, \bibinfo{person}{Moritz
  Hardt}, \bibinfo{person}{Toniann Pitassi}, \bibinfo{person}{Omer Reingold},
  {and} \bibinfo{person}{Rich Zemel}.} \bibinfo{year}{2011}\natexlab{}.
\newblock \showarticletitle{Fairness Through Awareness}.
\newblock \bibinfo{journal}{\emph{CoRR}}  \bibinfo{volume}{abs/1104.3913}
  (\bibinfo{date}{04} \bibinfo{year}{2011}).
\newblock
\urldef\tempurl%
\url{https://doi.org/10.1145/2090236.2090255}
\showDOI{\tempurl}


\bibitem[\protect\citeauthoryear{Ekstrand, Tian, Kazi, Mehrpouyan, and
  Kluver}{Ekstrand et~al\mbox{.}}{2018}]%
        {ekstrand}
\bibfield{author}{\bibinfo{person}{Michael Ekstrand}, \bibinfo{person}{Mucun
  Tian}, \bibinfo{person}{Mohammed Kazi}, \bibinfo{person}{Hoda Mehrpouyan},
  {and} \bibinfo{person}{Daniel Kluver}.} \bibinfo{year}{2018}\natexlab{}.
\newblock \showarticletitle{Exploring Author Gender in Book Rating and
  Recommendation}.
\newblock  (\bibinfo{date}{08} \bibinfo{year}{2018}).
\newblock
\urldef\tempurl%
\url{https://doi.org/10.1145/3240323.3240373}
\showDOI{\tempurl}


\bibitem[\protect\citeauthoryear{Hajian, Bonchi, and Castillo}{Hajian
  et~al\mbox{.}}{2016}]%
        {datamining11}
\bibfield{author}{\bibinfo{person}{Sara Hajian}, \bibinfo{person}{Francesco
  Bonchi}, {and} \bibinfo{person}{Carlos Castillo}.}
  \bibinfo{year}{2016}\natexlab{}.
\newblock \showarticletitle{Algorithmic Bias: From Discrimination Discovery to
  Fairness-aware Data Mining}. \bibinfo{pages}{2125--2126}.
\newblock
\urldef\tempurl%
\url{https://doi.org/10.1145/2939672.2945386}
\showDOI{\tempurl}


\bibitem[\protect\citeauthoryear{Hajian and Domingo-Ferrer}{Hajian and
  Domingo-Ferrer}{2013}]%
        {datamining12}
\bibfield{author}{\bibinfo{person}{Sara Hajian} {and} \bibinfo{person}{Josep
  Domingo-Ferrer}.} \bibinfo{year}{2013}\natexlab{}.
\newblock \showarticletitle{A Methodology for Direct and Indirect
  Discrimination Prevention in Data Mining}.
\newblock \bibinfo{journal}{\emph{IEEE Transactions on Knowledge and Data
  Engineering}} (\bibinfo{date}{07} \bibinfo{year}{2013}).
\newblock
\urldef\tempurl%
\url{https://doi.org/10.1109/TKDE.2012.72}
\showDOI{\tempurl}


\bibitem[\protect\citeauthoryear{Hajian, Domingo-Ferrer, and Farràs}{Hajian
  et~al\mbox{.}}{2014a}]%
        {datamining13}
\bibfield{author}{\bibinfo{person}{Sara Hajian}, \bibinfo{person}{Josep
  Domingo-Ferrer}, {and} \bibinfo{person}{Oriol Farràs}.}
  \bibinfo{year}{2014}\natexlab{a}.
\newblock \showarticletitle{Generalization-based privacy preservation and
  discrimination prevention in data publishing and mining}.
\newblock \bibinfo{journal}{\emph{Data Mining and Knowledge Discovery}}
  (\bibinfo{date}{09} \bibinfo{year}{2014}).
\newblock
\urldef\tempurl%
\url{https://doi.org/10.1007/s10618-014-0346-1}
\showDOI{\tempurl}


\bibitem[\protect\citeauthoryear{Hajian, Domingo-Ferrer, Monreale, Pedreschi,
  and Giannotti}{Hajian et~al\mbox{.}}{2014b}]%
        {datamining14}
\bibfield{author}{\bibinfo{person}{Sara Hajian}, \bibinfo{person}{Josep
  Domingo-Ferrer}, \bibinfo{person}{Anna Monreale}, \bibinfo{person}{Dino
  Pedreschi}, {and} \bibinfo{person}{Fosca Giannotti}.}
  \bibinfo{year}{2014}\natexlab{b}.
\newblock \showarticletitle{Discrimination- and privacy-aware patterns}.
\newblock \bibinfo{journal}{\emph{Data Mining and Knowledge Discovery}}
  \bibinfo{volume}{29} (\bibinfo{date}{11} \bibinfo{year}{2014}).
\newblock
\urldef\tempurl%
\url{https://doi.org/10.1007/s10618-014-0393-7}
\showDOI{\tempurl}


\bibitem[\protect\citeauthoryear{Herlocker, Konstan, Terveen, and
  Riedl}{Herlocker et~al\mbox{.}}{2004}]%
        {herlocker}
\bibfield{author}{\bibinfo{person}{Jonathan~L. Herlocker},
  \bibinfo{person}{Joseph~A. Konstan}, \bibinfo{person}{Loren~G. Terveen},
  {and} \bibinfo{person}{John~T. Riedl}.} \bibinfo{year}{2004}\natexlab{}.
\newblock \showarticletitle{Evaluating Collaborative Filtering Recommender
  Systems}.
\newblock \bibinfo{journal}{\emph{ACM Trans. Inf. Syst.}} \bibinfo{volume}{22},
  \bibinfo{number}{1} (\bibinfo{date}{Jan.} \bibinfo{year}{2004}),
  \bibinfo{pages}{5–53}.
\newblock
\showISSN{1046-8188}
\urldef\tempurl%
\url{https://doi.org/10.1145/963770.963772}
\showDOI{\tempurl}


\bibitem[\protect\citeauthoryear{Hurley and Zhang}{Hurley and Zhang}{2011}]%
        {neilHurley}
\bibfield{author}{\bibinfo{person}{Neil Hurley} {and} \bibinfo{person}{Mi
  Zhang}.} \bibinfo{year}{2011}\natexlab{}.
\newblock \showarticletitle{Novelty and Diversity in Top-N Recommendation --
  Analysis and Evaluation}.
\newblock \bibinfo{journal}{\emph{ACM Trans. Internet Technol.}}
  \bibinfo{volume}{10}, \bibinfo{number}{4}, Article \bibinfo{articleno}{14}
  (\bibinfo{date}{March} \bibinfo{year}{2011}), \bibinfo{numpages}{30}~pages.
\newblock
\showISSN{1533-5399}
\urldef\tempurl%
\url{https://doi.org/10.1145/1944339.1944341}
\showDOI{\tempurl}


\bibitem[\protect\citeauthoryear{ISBNDB}{ISBNDB}{2 27}]%
        {isbndbAPI}
\bibfield{author}{\bibinfo{person}{ISBN Database~| ISBNDB}.}
  \bibinfo{year}{Accessed: 2021-02-27}\natexlab{}.
\newblock
\newblock
\urldef\tempurl%
\url{https://isbndb.com/isbn-database}
\showURL{%
\tempurl}


\bibitem[\protect\citeauthoryear{Kamiran, Calders, and Pechenizkiy}{Kamiran
  et~al\mbox{.}}{2010}]%
        {datamining18}
\bibfield{author}{\bibinfo{person}{Faisal Kamiran}, \bibinfo{person}{Toon
  Calders}, {and} \bibinfo{person}{Mykola Pechenizkiy}.}
  \bibinfo{year}{2010}\natexlab{}.
\newblock \showarticletitle{Discrimination Aware Decision Tree Learning}.
  \bibinfo{pages}{869--874}.
\newblock
\urldef\tempurl%
\url{https://doi.org/10.1109/ICDM.2010.50}
\showDOI{\tempurl}


\bibitem[\protect\citeauthoryear{Kamiran, Karim, and Zhang}{Kamiran
  et~al\mbox{.}}{2012}]%
        {datamining19}
\bibfield{author}{\bibinfo{person}{Faisal Kamiran}, \bibinfo{person}{Asim
  Karim}, {and} \bibinfo{person}{Xiangliang Zhang}.}
  \bibinfo{year}{2012}\natexlab{}.
\newblock \showarticletitle{Decision Theory for Discrimination-Aware
  Classification}.
\newblock \bibinfo{journal}{\emph{Proceedings - IEEE International Conference
  on Data Mining, ICDM}}, \bibinfo{pages}{924--929}.
\newblock
\showISBNx{978-1-4673-4649-8}
\urldef\tempurl%
\url{https://doi.org/10.1109/ICDM.2012.45}
\showDOI{\tempurl}


\bibitem[\protect\citeauthoryear{Knijnenburg, Willemsen, gantner, and
  newell}{Knijnenburg et~al\mbox{.}}{2012}]%
        {knijnenburg}
\bibfield{author}{\bibinfo{person}{Bart Knijnenburg}, \bibinfo{person}{Martijn
  Willemsen}, \bibinfo{person}{soncu gantner}, {and} \bibinfo{person}{newell}.}
  \bibinfo{year}{2012}\natexlab{}.
\newblock \showarticletitle{Explaining the user experience of recommender
  systems}.
\newblock \bibinfo{journal}{\emph{User Modeling and User-Adapted Interaction}}
  \bibinfo{volume}{22} (\bibinfo{date}{10} \bibinfo{year}{2012}),
  \bibinfo{pages}{441--504}.
\newblock
\urldef\tempurl%
\url{https://doi.org/10.1007/s11257-011-9118-4}
\showDOI{\tempurl}


\bibitem[\protect\citeauthoryear{Leavy, Meaney, Wade, and Greene}{Leavy
  et~al\mbox{.}}{2020}]%
        {contentBasedFiltering}
\bibfield{author}{\bibinfo{person}{Susan Leavy}, \bibinfo{person}{Gerardine
  Meaney}, \bibinfo{person}{Karen Wade}, {and} \bibinfo{person}{Derek Greene}.}
  \bibinfo{year}{2020}\natexlab{}.
\newblock \showarticletitle{Mitigating Gender Bias in Machine Learning Data
  Sets}.
\newblock  (\bibinfo{date}{05} \bibinfo{year}{2020}).
\newblock


\bibitem[\protect\citeauthoryear{Mancuhan and Clifton}{Mancuhan and
  Clifton}{2014}]%
        {datamining26}
\bibfield{author}{\bibinfo{person}{Koray Mancuhan} {and} \bibinfo{person}{Chris
  Clifton}.} \bibinfo{year}{2014}\natexlab{}.
\newblock \showarticletitle{Combating discrimination using Bayesian networks}.
\newblock \bibinfo{journal}{\emph{Artificial Intelligence and Law}}
  \bibinfo{volume}{22} (\bibinfo{date}{06} \bibinfo{year}{2014}).
\newblock
\urldef\tempurl%
\url{https://doi.org/10.1007/s10506-014-9156-4}
\showDOI{\tempurl}


\bibitem[\protect\citeauthoryear{Mansoury, Abdollahpouri, Smith, Dehpanah,
  Pechenizkiy, and Mobasher}{Mansoury et~al\mbox{.}}{2020}]%
        {potentialFactors}
\bibfield{author}{\bibinfo{person}{Masoud Mansoury}, \bibinfo{person}{Himan
  Abdollahpouri}, \bibinfo{person}{Jessie Smith}, \bibinfo{person}{Arman
  Dehpanah}, \bibinfo{person}{Mykola Pechenizkiy}, {and}
  \bibinfo{person}{Bamshad Mobasher}.} \bibinfo{year}{2020}\natexlab{}.
\newblock \showarticletitle{Investigating Potential Factors Associated with
  Gender Discrimination in Collaborative Recommender Systems}.
\newblock  (\bibinfo{date}{02} \bibinfo{year}{2020}).
\newblock


\bibitem[\protect\citeauthoryear{Neve and Palomares}{Neve and
  Palomares}{2019}]%
        {neve2019latent}
\bibfield{author}{\bibinfo{person}{James Neve} {and} \bibinfo{person}{Ivan
  Palomares}.} \bibinfo{year}{2019}\natexlab{}.
\newblock \showarticletitle{Latent factor models and aggregation operators for
  collaborative filtering in reciprocal recommender systems}. In
  \bibinfo{booktitle}{\emph{Proceedings of the 13th ACM Conference on
  Recommender Systems}}. \bibinfo{pages}{219--227}.
\newblock


\bibitem[\protect\citeauthoryear{Ni, Li, and McAuley}{Ni et~al\mbox{.}}{2019}]%
        {amazon}
\bibfield{author}{\bibinfo{person}{Jianmo Ni}, \bibinfo{person}{Jiacheng Li},
  {and} \bibinfo{person}{Julian McAuley}.} \bibinfo{year}{2019}\natexlab{}.
\newblock \showarticletitle{Justifying Recommendations using Distantly-Labeled
  Reviews and Fine-Grained Aspects}. \bibinfo{pages}{188--197}.
\newblock
\urldef\tempurl%
\url{https://doi.org/10.18653/v1/D19-1018}
\showDOI{\tempurl}


\bibitem[\protect\citeauthoryear{OpenLibrary}{OpenLibrary}{3 02}]%
        {openLibraryAPI}
\bibfield{author}{\bibinfo{person}{Developers API~| OpenLibrary}.}
  \bibinfo{year}{Accessed: 2021-03-02}\natexlab{}.
\newblock
\newblock
\urldef\tempurl%
\url{https://openlibrary.org/developers/api}
\showURL{%
\tempurl}


\bibitem[\protect\citeauthoryear{Pedreschi, Ruggieri, and Turini}{Pedreschi
  et~al\mbox{.}}{2008}]%
        {datamining29}
\bibfield{author}{\bibinfo{person}{Dino Pedreschi}, \bibinfo{person}{Salvatore
  Ruggieri}, {and} \bibinfo{person}{Franco Turini}.}
  \bibinfo{year}{2008}\natexlab{}.
\newblock \showarticletitle{Discrimination-aware data mining}.
\newblock \bibinfo{journal}{\emph{Proceedings of the ACM SIGKDD International
  Conference on Knowledge Discovery and Data Mining}},
  \bibinfo{pages}{560--568}.
\newblock
\urldef\tempurl%
\url{https://doi.org/10.1145/1401890.1401959}
\showDOI{\tempurl}


\bibitem[\protect\citeauthoryear{Pedreschi, Ruggieri, and Turini}{Pedreschi
  et~al\mbox{.}}{2009}]%
        {datamining28}
\bibfield{author}{\bibinfo{person}{Dino Pedreschi}, \bibinfo{person}{Salvatore
  Ruggieri}, {and} \bibinfo{person}{Franco Turini}.}
  \bibinfo{year}{2009}\natexlab{}.
\newblock \showarticletitle{Measuring Discrimination in Socially-Sensitive
  Decision Records}. \bibinfo{pages}{581--592}.
\newblock
\urldef\tempurl%
\url{https://doi.org/10.1137/1.9781611972795.50}
\showDOI{\tempurl}


\bibitem[\protect\citeauthoryear{Ruggieri, Hajian, Kamiran, and Zhang}{Ruggieri
  et~al\mbox{.}}{2014}]%
        {datamining30}
\bibfield{author}{\bibinfo{person}{Salvatore Ruggieri}, \bibinfo{person}{Sara
  Hajian}, \bibinfo{person}{Faisal Kamiran}, {and} \bibinfo{person}{Xiangliang
  Zhang}.} \bibinfo{year}{2014}\natexlab{}.
\newblock \showarticletitle{Anti-discrimination Analysis Using Privacy Attack
  Strategies}.
\newblock
\showISBNx{978-3-662-44850-2}
\urldef\tempurl%
\url{https://doi.org/10.1007/978-3-662-44851-9_44}
\showDOI{\tempurl}


\bibitem[\protect\citeauthoryear{Ruggieri, Pedreschi, and Turini}{Ruggieri
  et~al\mbox{.}}{2010}]%
        {datamining31}
\bibfield{author}{\bibinfo{person}{Salvatore Ruggieri}, \bibinfo{person}{Dino
  Pedreschi}, {and} \bibinfo{person}{Franco Turini}.}
  \bibinfo{year}{2010}\natexlab{}.
\newblock \showarticletitle{Data Mining for Discrimination Discovery}.
\newblock \bibinfo{journal}{\emph{TKDD}}  \bibinfo{volume}{4}
  (\bibinfo{date}{05} \bibinfo{year}{2010}).
\newblock
\urldef\tempurl%
\url{https://doi.org/10.1145/1754428.1754432}
\showDOI{\tempurl}


\bibitem[\protect\citeauthoryear{Shakespeare, Porcaro, Gómez, and
  Castillo}{Shakespeare et~al\mbox{.}}{2020}]%
        {musicRecommender}
\bibfield{author}{\bibinfo{person}{Dougal Shakespeare},
  \bibinfo{person}{Lorenzo Porcaro}, \bibinfo{person}{Emilia Gómez}, {and}
  \bibinfo{person}{Carlos Castillo}.} \bibinfo{year}{2020}\natexlab{}.
\newblock \showarticletitle{Exploring Artist Gender Bias in Music
  Recommendation}.
\newblock  (\bibinfo{date}{09} \bibinfo{year}{2020}).
\newblock


\bibitem[\protect\citeauthoryear{Shani and Gunawardana}{Shani and
  Gunawardana}{2011}]%
        {shani}
\bibfield{author}{\bibinfo{person}{Guy Shani} {and} \bibinfo{person}{Asela
  Gunawardana}.} \bibinfo{year}{2011}\natexlab{}.
\newblock \bibinfo{booktitle}{\emph{Evaluating Recommendation Systems}}.
  Vol.~\bibinfo{volume}{12}.
\newblock \bibinfo{pages}{257--297}.
\newblock
\urldef\tempurl%
\url{https://doi.org/10.1007/978-0-387-85820-3_8}
\showDOI{\tempurl}


\bibitem[\protect\citeauthoryear{Thanh, Ruggieri, and Turini}{Thanh
  et~al\mbox{.}}{2011}]%
        {datamining25}
\bibfield{author}{\bibinfo{person}{Binh Thanh}, \bibinfo{person}{Salvatore
  Ruggieri}, {and} \bibinfo{person}{Franco Turini}.}
  \bibinfo{year}{2011}\natexlab{}.
\newblock \showarticletitle{k-NN as an Implementation of Situation Testing for
  Discrimination Discovery and Prevention}.
\newblock \bibinfo{journal}{\emph{Proceedings of the ACM SIGKDD International
  Conference on Knowledge Discovery and Data Mining}},
  \bibinfo{pages}{502--510}.
\newblock
\urldef\tempurl%
\url{https://doi.org/10.1145/2020408.2020488}
\showDOI{\tempurl}


\bibitem[\protect\citeauthoryear{the gender of~a name}{the gender of~a name}{03
  5}]%
        {genderize}
\bibfield{author}{\bibinfo{person}{Genderize.io |~Deteremine the gender of~a
  name}.} \bibinfo{year}{Accessed: 2021-03-5}\natexlab{}.
\newblock
\newblock
\urldef\tempurl%
\url{https://genderize.io/}
\showURL{%
\tempurl}


\bibitem[\protect\citeauthoryear{Tsintzou, Pitoura, and Tsaparas}{Tsintzou
  et~al\mbox{.}}{2018}]%
        {metric}
\bibfield{author}{\bibinfo{person}{Virginia Tsintzou},
  \bibinfo{person}{Evaggelia Pitoura}, {and} \bibinfo{person}{Panayiotis
  Tsaparas}.} \bibinfo{year}{2018}\natexlab{}.
\newblock \showarticletitle{Bias Disparity in Recommendation Systems}.
\newblock  (\bibinfo{date}{11} \bibinfo{year}{2018}).
\newblock


\bibitem[\protect\citeauthoryear{Valcarce, Bellog{\'\i}n, Parapar, and
  Castells}{Valcarce et~al\mbox{.}}{2020}]%
        {valcarce2020assessing}
\bibfield{author}{\bibinfo{person}{Daniel Valcarce}, \bibinfo{person}{Alejandro
  Bellog{\'\i}n}, \bibinfo{person}{Javier Parapar}, {and}
  \bibinfo{person}{Pablo Castells}.} \bibinfo{year}{2020}\natexlab{}.
\newblock \showarticletitle{Assessing ranking metrics in top-N recommendation}.
\newblock \bibinfo{journal}{\emph{Information Retrieval Journal}}
  \bibinfo{volume}{23} (\bibinfo{year}{2020}), \bibinfo{pages}{411--448}.
\newblock


\bibitem[\protect\citeauthoryear{Zemel, Wu, Swersky, Pitassi, and Dwork}{Zemel
  et~al\mbox{.}}{2013}]%
        {datamining33}
\bibfield{author}{\bibinfo{person}{R. Zemel}, \bibinfo{person}{Y. Wu},
  \bibinfo{person}{K. Swersky}, \bibinfo{person}{T. Pitassi}, {and}
  \bibinfo{person}{C. Dwork}.} \bibinfo{year}{2013}\natexlab{}.
\newblock \showarticletitle{Learning fair representations}.
\newblock \bibinfo{journal}{\emph{30th International Conference on Machine
  Learning, ICML 2013}} (\bibinfo{date}{01} \bibinfo{year}{2013}),
  \bibinfo{pages}{1362--1370}.
\newblock


\bibitem[\protect\citeauthoryear{Ziegler, McNee, Konstan, and Lausen}{Ziegler
  et~al\mbox{.}}{2005a}]%
        {zeigler}
\bibfield{author}{\bibinfo{person}{Cai-Nicolas Ziegler},
  \bibinfo{person}{Sean~M. McNee}, \bibinfo{person}{Joseph~A. Konstan}, {and}
  \bibinfo{person}{Georg Lausen}.} \bibinfo{year}{2005}\natexlab{a}.
\newblock \showarticletitle{Improving Recommendation Lists through Topic
  Diversification} \emph{(\bibinfo{series}{WWW '05})}.
  \bibinfo{publisher}{Association for Computing Machinery},
  \bibinfo{address}{New York, NY, USA}, \bibinfo{pages}{22–32}.
\newblock
\showISBNx{1595930469}
\urldef\tempurl%
\url{https://doi.org/10.1145/1060745.1060754}
\showDOI{\tempurl}


\bibitem[\protect\citeauthoryear{Ziegler, McNee, Konstan, and Lausen}{Ziegler
  et~al\mbox{.}}{2005b}]%
        {bookCrossing}
\bibfield{author}{\bibinfo{person}{Cai-Nicolas Ziegler},
  \bibinfo{person}{Sean~M. McNee}, \bibinfo{person}{Joseph~A. Konstan}, {and}
  \bibinfo{person}{Georg Lausen}.} \bibinfo{year}{2005}\natexlab{b}.
\newblock \showarticletitle{Improving Recommendation Lists through Topic
  Diversification}. In \bibinfo{booktitle}{\emph{Proceedings of the 14th
  International Conference on World Wide Web}} (Chiba, Japan)
  \emph{(\bibinfo{series}{WWW '05})}. \bibinfo{publisher}{Association for
  Computing Machinery}, \bibinfo{address}{New York, NY, USA},
  \bibinfo{pages}{22–32}.
\newblock
\showISBNx{1595930469}
\urldef\tempurl%
\url{https://doi.org/10.1145/1060745.1060754}
\showDOI{\tempurl}


\end{thebibliography}
\end{document}